\documentclass[11pt]{article}
\usepackage{latexsym,graphics,graphicx,amssymb,amsmath, amstext, amsfonts, amsthm, pifont, color,  subfig,times}
\usepackage{hyperref}
\usepackage{boxedminipage, wrapfig, picins}

\newtheorem{defn}{Definition}

\newtheorem{prop}{Proposition}

\newtheorem{note}{Remark}

\newcommand{\bc}{\begin{center}}
\newcommand{\ec}{\end{center}}
\newcommand{\bfl}{\begin{flushleft}}
\newcommand{\efl}{\end{flushleft}}

\newcommand{\beqa}{\begin{eqnarray}}
\newcommand{\eeqa}{\end{eqnarray}}

\newcommand{\beqan}{\begin{eqnarray*}}
\newcommand{\eeqan}{\end{eqnarray*}}
\newcommand{\beq}{\begin{equation}}
\newcommand{\eeq}{\end{equation}}
\newcommand{\beit}{\begin{itemize}}
\newcommand{\eeit}{\end{itemize}}

\newcommand{\mc}{\mathcal}
\newcommand{\mb}{\mathbb}

\newcommand{\bit}{\begin{itemize}}
\newcommand{\eit}{\end{itemize}}
\newcommand{\ben}{\begin{enumerate}}
\newcommand{\een}{\end{enumerate}}

\newcommand{\bdefn}{\begin{defn}}
\newcommand{\edefn}{\end{defn}}
\newcommand{\bnote}{\begin{note}}
\newcommand{\enote}{\end{note}}

\newcommand{\blem}{\begin{lemma}}
\newcommand{\elem}{\end{lemma}}
\newcommand{\bthm}{\begin{theorem}}
\newcommand{\ethm}{\end{theorem}}
\newcommand{\bpf}{\begin{proof}}
\newcommand{\epf}{\end{proof}}
\newcommand{\bcor}{\begin{corollary}}
\newcommand{\ecor}{\end{corollary}}
\newcommand{\bprop}{\begin{proposition}}
\newcommand{\eprop}{\end{proposition}}

\newtheorem{proposition}{Proposition}
\newtheorem{lemma}{Lemma}
\newtheorem{theorem}{Theorem}

\newtheorem{remark}{Remark}
\newtheorem{corollary}{Corollary}
\newcommand{\dvin}{d_v^{\text{in}}}

\newcommand{\Fvin}{\mc{F}_v^{-}}

\newcommand{\In}{\delta^{-}}
\newcommand{\Out}{\delta^{+}}
\newcommand{\mcP}{\mathcal{P}}

\newcommand{\st}{(s_i,t_i)}

\newcommand{\bigO}{\mathcal{O}}

\newcommand{\cP}{\mathcal{P}}
\newcommand{\cT}{\mathcal{T}}
\newcommand{\bx}{{\bf x}}
\newcommand{\bd}{{\bf d}}

\newcommand{\trho}{\tilde{\rho}}
\newcommand{\hrho}{\hat{\rho}}
\newcommand{\dist}{\text{dist}}
\newcommand{\avgd}{\text{avgd}}
\newcommand{\vol}{\text{{\rm vol}}}
\newcommand{\sep}{\text{sep}}

\newcommand{\etal}{et al.\ }

\makeatletter

\title{Multicommodity Flows and Cuts in Polymatroidal
  Networks\footnote{An extended abstract will appear in
{\em Proc.\ of the Innovations in Theoretical Computer Science Conference ((ITCS)}, January 2012.}}

\author{Chandra Chekuri\thanks{Supported in part by NSF grant
CCF-1016684.}
\and
Sreeram Kannan\thanks{Supported in part by NSF grants CCF 1017430 and
CNS 0721652.}
\and
Adnan Raja\thanks{Supported in part by NSF grants CCF 1017430 and
CNS 0721652.}
\and
Pramod Viswanath\thanks{Supported in part by NSF grants CCF 1017430 and
CNS 0721652.}
}

\date{University of Illinois, Urbana-Champaign, IL  61801 \\
{\tt \{chekuri,kannan1,araja2,pramodv\}@illinois.edu} \\ ~\\ \today }
\begin{document}
\maketitle

\thispagestyle{empty}
\begin{abstract}
  We consider multicommodity flow and cut problems in {\em
    polymatroidal} networks where there are submodular capacity
  constraints on the edges incident to a node. Polymatroidal networks
  were introduced by Lawler and Martel \cite{LawlerMartel} and Hassin
  \cite{Hassin} in the single-commodity setting and are closely
  related to the submodular flow model of Edmonds and Giles
  \cite{EdmondsGiles}; the well-known maxflow-mincut theorem
  holds in this more general setting. Polymatroidal networks for
  the multicommodity case have not, as far as the authors are aware,
  been previously explored. Our work is primarily motivated by
  applications to information flow in wireless networks. We also
  consider the notion of undirected polymatroidal networks and observe
  that they provide a natural way to generalize flows and cuts
  in edge and node capacitated undirected networks.

  We establish poly-logarithmic flow-cut gap results in several
  scenarios that have been previously considered in the standard
  network flow models where capacities are on the edges or nodes
  \cite{LeightonRao,LLR,GVY,KleinPRT97,FeigeHL06}. Our results from a
  preliminary version have already found applications in wireless
  network information flow \cite{KannanRV11,KannanV11} and we
  anticipate more in the future.  On the technical side our key tools
  are the formulation and analysis of the dual of the flow relaxations
  via continuous extensions of submodular functions, in particular the
  Lov\'asz extension. For directed graphs we rely on a simple yet
  useful reduction from polymatroidal networks to standard networks.
  For undirected graphs we rely on the interplay between the Lov\'asz
  extension of a submodular function and line embeddings with low
  average distortion introduced by Matousek and Rabinovich \cite{MatousekR01};
  this connection is inspired by, and generalizes, the work of Feige,
  Hajiaghayi and Lee \cite{FeigeHL06} on node-capacitated
  multicommodity flows and cuts.  The applicability of embeddings to
  flow-cut gaps in polymatroidal networks is of independent
  mathematical interest.
\end{abstract}

\newpage
\setcounter{page}{1}
\section{Introduction}\label{sec:intro}
Consider a communication network represented by a directed graph $G =
(V,E)$. In the so-called edge-capacitated scenario, each edge $e$ has
an associated capacity $c(e)$ that limits the information flowing on
it. We consider a more general network model called the {\em
  polymatroidal network} introduced by Lawler and Martel
\cite{LawlerMartel} and independently by Hassin \cite{Hassin}.  This
model is closely related to the submodular flow model introduced by
Edmonds and Giles \cite{EdmondsGiles}. Both models capture as special
cases, single-commodity $s$-$t$ flows in edge-capacitated directed
networks, and polymatroid intersection, hence their
importance. Moreover the models are known to be equivalent (see
Chapter 60 in \cite{Schrijver}, in particular Section 60.3b). The
polymatroidal network flow model is more directly and intuitively
related to standard network flows and one can easily generalize it
to the multicommodity setting which is the focus in this paper.

The polymatroidal network flow model differs from the standard network
flow model in the following way. Consider a node $v$ in a directed
graph $G$ and let $\delta^-_G(v)$ be the set of edges in to $v$ and
$\delta^+_G(v)$ be the set of edges out of $v$.  In the standard model
each edge $(u,v)$ has a non-negative capacity $c(u,v)$ that is
independent of other edges.  In the polymatroidal network for each
node $v$ there are two associated submodular functions (in fact
polymatroids\footnote{A set function $f:2^N \rightarrow \mathbb{R}$
  over a finite ground set $N$ is submodular iff $f(A) + f(B) \ge f(A
  \cap B) + f(A \cup B)$ for all $A, B \subseteq N$; equivalently
  $f(A\cup \{i\}) -f(A) \ge f(B \cup \{i\}) - f(B)$ for all $A \subset
  B$ and $i \not \in A$. It is monotone if $f(A) \le f(B)$ for all $A
  \subset B$.  In this paper a polymatroid refers to a non-negative monotone
  submodular function with $f(\emptyset) = 0$.})  $\rho^-_v$ and
$\rho^+_v$ which impose joint capacity constraints on the edges in
$\delta^{-}_G(v)$ and $\delta^+_G(v)$ respectively. That is, for any
set of edges $S \subseteq \delta^-_G(v)$, the total capacity available
on the edges in $S$ is constrained to be at most $\rho^-_v(S)$,
similarly for $\delta^+_G(v)$. Note that an edge $(u,v)$ is influenced
by $\rho^+_u$ and $\rho^-_v$.  Lawler and Martel considered the
problem of finding a maximum $s$-$t$ flow in this model. The results
in \cite{LawlerMartel,Hassin} show that various important properties
that hold for $s$-$t$ flows in standard networks generalize to
polymatroid networks; these include the classical maxflow-mincut
theorem of Ford and Fulkerson (and Menger) and the existence of an
integer valued maximum flow when capacities are integral.

The original motivation for the Lawler-Martel model came from an
application to a scheduling problem \cite{Martel82}. More
recently, there have been several applications of polymatroid
network flows, (and submodular flows) and their generalizations such as
linking systems \cite{Schrijver-Thesis}, to information flow in
wireless networks
\cite{ADT,AmaFrag,YazdiSavari,GoemansIZ,RajaV11,KannanRV11}.
A node in a wireless network communicates with several nodes
over a broadcast medium and hence the channels interfere with each
other; this imposes joint capacity constraints on the channels.
Several interference scenarios of interest can be modeled by
submodular functions. Most of the work on this topic so far has
focused on the case of a single source.  In this paper we consider
{\em multicommodity} flows and cuts in polymatroidal networks where
several source-sink pairs $(s_1,t_1), (s_2,t_2),\ldots, (s_k,t_k)$
share the capacity of the network.  In the communications literature
this is referred to as the multiple unicast setting.  Our primary
motivation is applications to (wireless) network information flow; see
companion papers \cite{KannanRV11,KannanV11} that build on results of
this paper. Another motivation is to understand the extent to which
techniques and results that were developed for multicommodity
flows and cuts in standard networks generalize to polymatroidal
networks. We note that polymatroidal networks allow for a common
treatment of edge and node capacities; an advantage is that one can
define cuts with respect to edge removals while the cost is based on
nodes. As far as we are aware, multicommodity flows and cuts in
polymatroidal networks have not been studied previously.

\medskip
\noindent
{\bf Flow-cut gaps in polymatroidal networks:}
The main focus of this paper is understanding multicommodity flow-cut
gaps in polymatroidal networks. In communication networks cuts can be
used to information theoretically upper bound achievable rates while
flows allow one to develop lower bounds on achievable rates by
combining a variety of routing and coding schemes. Flow-cut gaps are
of therefore of much interest. Unlike the case of single-commodity
flows where maximum flow is equal to minimum cut, it is well-known
that even in standard edge-capacitated networks no tight min-max
result holds when the number of source-sink pairs is three or more
(two or more in case of directed graphs). See \cite{Schrijver} for
some special cases where min-max results do hold. Flow-cut gap
results have been extensively studied in theoretical computer science
starting with the seminal work of Leighton and Rao
\cite{LeightonRao}. The initial motivation was approximation
algorithms for cut and separator problems that are NP-Hard. There has
been much subsequent work with a tight bound of $O(\log k)$
established for flow-cut gaps in undirected graphs in a variety of
settings \cite{GVY,LLR,AumannR98,FeigeHL06}.  It has also been shown
that strong lower bounds exist for flow-cut gaps in directed graphs;
for instance the gap is $O(\min\{k, n^\delta\})$ between the maximum
concurrent flow and the sparsest cut \cite{SaksSZ04,ChuzhoyK07} where
$\delta$ is a fixed constant.  However poly-logarithmic upper bounds
on the gaps are known for the case of symmetric demands in directed
graphs âÃ\cite{KleinPRT97,ENRS}.  Motivated by the above
positive and negative results we focus on those cases where
poly-logarithmic flow-cut gaps have been established.  We show that
several of these gap results extend to polymatroid networks. Our
results and techniques lead to new approximation algorithms for cut
problems in polymatroidal networks which could have
future applications. However, in this paper we restrict our attention to
quantifying flow-cut gaps.

\medskip
\noindent {\bf Bidirected and undirected polymatroidal networks:} As
we mentioned already, strong lower bounds exist on flow-cut gaps for
directed networks. Positive results in the form of poly-logarithmic
upper bounds on flow-cut gaps for standard networks hold when the
demands are symmetric or when the supply graph is undirected. A
natural model for wireless networks is the {\em bidirected}
polymatroidal network. For two nodes $u$ and $v$ it is a reasonable
approximation to assume that the channel from $u$ to $v$ is similar to
that from $v$ to $u$; hence one can assume that the underlying graph
$G$ is bidirected in that if the edge $(u,v)$ is present then so is
$(v,u)$. Moreover, we assume that for any node $v$ and $S \subseteq
\delta^-(v)$, $\rho^{-}_v(S) = \rho^+_v(S')$ where $S'\subseteq
\delta^+(v)$ is the set of edges that correspond to the reverse of the
edges in $S$.  Within a factor of $2$ bidirected polymatroidal
networks can be approximated by {\em undirected} polymatroidal
networks: we have an undirected graph $G$ and for each node $v$ a
single polymatroid $\rho_v$ that constrains the capacity of the edges
$\delta_G(v)$, the set of edges incident to $v$. The main advantage of
undirected polymatroid networks is that we can use existing tools and
ideas from metric embeddings to understand flow-cut gap
results. Undirected polymatroidal networks have not been considered
previously. We observe that they allow a
natural way to capture both edge and node-capacitated flows in
undirected graphs. To capture node-capacitated flows we set $\rho_v(S)
= 2c(v)$ for all $\emptyset \neq S \subseteq \delta(v)$ where $c(v)$
is the capacity of $v$ \footnote{The factor of $2$ is needed since a
  flow path $p$ through an internal node $v$ uses two edges. On the other
  hand it is not needed for the sources and sinks. This
  technical issue is a minor inconvenience with undirected
  polymatroidal networks; we note that this also arises in treating
  node-capacitated multicommodity flows \cite{FeigeHL06}.}. We mention
an advantage of using polymatroidal networks even when considering the
special case of node-capacitated flows and cuts: one can define cuts
with respect to edges even though the cost is on the nodes. This is in
fact quite natural and simplifies certain aspects of the algorithms in
\cite{FeigeHL06}.

\subsection{Overview of  results and technical ideas}
We do a systematic study of flow-cut gaps in multicommodity
polymatroidal networks, both directed and undirected. Let $G=(V,E)$ be
a polymatroidal network on $n$ nodes with $k$ source-sink pairs
$(s_1,t_1), \ldots,(s_k,t_k)$.  We consider two flow
problems and their corresponding cut problems: (i) maximum
throughput flow and multicut (ii) maximum concurrent flow and sparsest cut.
Our high-level results are summarized below.
\begin{itemize}
\item For directed networks we show a reduction based on the dual that
  establishes a correspondence between flow-cut gaps in polymatroidal
  networks and the standard edge-capacitated networks. This allows us to
  obtain poly-logarithmic upper bounds for flow-cut gaps in directed
  polymatroidal networks with {\em symmetric} demands via results in
  \cite{KleinPRT97,ENRS} for both throughput flow and concurrent flow.
  In particular we obtain an $O(\min\{\log^3 k, \log^2 n \log \log n\})$
  gap between the maximum concurrent flow and sparsest cut. The reduction
  is applicable only to directed graphs.
\item We show that line embeddings with low average distortion
  \cite{MatousekR01,Rabinovich03} lead to upper bounds on flow-cut gaps in
  polymatroidal networks --- this connection is inspired by the work in
  \cite{FeigeHL06} for node-capacitated flows. For undirected
  polymatroidal networks this leads to an optimal $O(\log k)$ gap between
  maximum concurrent flow and sparsest cut. We also obtain an optimal $O(\log
  k)$ gap between throughput flow and multicut. These imply
  corresponding results for bidirected networks,  which have
  already found applications \cite{KannanRV11,KannanV11}. As in \cite{FeigeHL06}
  the embedding connection can be exploited to obtain improved
  approximation algorithms for certain separator problems by
  exploiting graph structure \cite{KPR} or by using stronger
  relaxations via semi-definite programming (and associated embedding
  theorems) \cite{ARV,AgarwalAC07}; we defer these improvements to a
  later version.
\end{itemize}

Most of the literature on multicommodity flow-cut gaps is based on
analyzing the dual of the linear program for the flow which can be
viewed as a fractional relaxation for the corresponding cut problem.
The gap is established by showing the existence of an integral cut
within some factor of the relaxation. For standard edge and
node-capacitated network flows the dual linear program has length
variables on the edges which induce distances on the nodes. The
situation is more involved in polymatroidal networks, in particular,
the definition of the cost of a cut is some what complex and is
discussed in more detail in Section~\ref{sec:cuts}. Our starting point
is the use of the Lov\'asz extension of a submodular function
\cite{Lovasz83} to cleanly rewrite the dual of the flow linear
programs.  This simplifies the constraint structure of the dual at the
expense of making the objective a convex function. However, we are
able to exploit properties of the Lov\'asz extension in several ways
to obtain our results.  Our techniques give two new dual-based proofs
of the maxflow-mincut theorem for single commodity polymatroid
networks that was first established by Lawler and Martel
algorithmically \cite{LawlerMartel} via an augmenting path based
approach. We believe that the applicability of embedding based methods
for polymatroidal networks is of independent mathematical interest.

For the most part we ignore algorithmic issues in this paper although
all the flow-cut gap results lead to efficient algorithms for finding
approximate cuts.

\section{Multicommodity Flows and Cuts in Polymatroidal Networks}
\label{sec:flow-cut-defns}
We let $G=(V,E)$ represent a graph whether directed or undirected. We
use $(u,v)$ for an ordered pair of nodes and $uv$ to denote an
unordered pair. In a directed graph $G$, for a given node $v$,
$\delta^-_G(v)$ and $\delta^+_G(v)$ denote the set of incoming and
outgoing edges at $v$. In undirected graphs we use $\delta_G(v)$ to
denote the set of edges incident to $v$. We omit the subscript $G$ if
it is clear from the context. In addition to the graph the input
consists of a set of $k$ source-sink pairs
$(s_1,t_1),\ldots,(s_k,t_k)$ that wish to communicate independently
and share the network capacity.

In a directed polymatroidal network each node $v \in V$ has two
associated polymatroids $\rho^-_v$ and $\rho^+_v$ with ground sets as
$\delta^-(v)$ and $\delta^+(v)$ respectively. These functions
constrain the joint capacity on the edges incident to $v$ as
follows. If $S \subseteq \delta^-(v)$ then $\rho^-_v(S)$ upper bounds
the total capacity of the edges in $S$ and similarly if $S \subseteq
\delta^+(v)$ then $\rho^+_v(S)$ upper bounds the total capacity of the
edges in $S$. We assume that the functions $\rho^-_v(\cdot),\rho^+_v(\cdot) \quad v \in V$ are provided via value oracles.
In undirected polymatroidal graphs we have a single function
$\rho_v(\cdot)$ at a node $v$ that constrains the capacity of
the edges incident to $v$.
Continuous extensions of submodular functions, namely the Lov\'asz
extension \cite{Lovasz83} and the convex closure, are important
technical tools in interpreting and analyzing the duals of the linear
programs for multicommodity flow in the polymatroid setting. We
discuss these in Section~\ref{sec:cuts}.
We first discuss the two flow problems of interest, namely maximum
throughput flow and the maximum concurrent flow.

\subsection{Flows \label{sec:flows}}
A multicommodity flow for a given collection of $k$ source-sink pairs
$(s_1,t_1),\ldots,(s_k,t_k)$ consists of $k$ separate single-commodity
flows, one for each pair $(s_i,t_i)$. The flow for the $i$'th commodity
can either be viewed as an edge-based flow $f_i:E \rightarrow \mathbb{R}_+$
or as a path-based flow $f_i: \cP_i \rightarrow \mathbb{R}_+$ where
$\cP_i$ is the set of all simple paths between $s_i$ and $t_i$ in $G$.
We prefer the path-based flow since it is more convenient for treating
directed and undirected graphs in a unified fashion, and also to
write the linear programs for flows and cuts in a more intuitive fashion.
However, it is easier to argue polynomial-time solvability of
the linear programs via edge-based flows.
Given path-based flows $f_i$, $i=1,\ldots,k$ for the $k$ source-sink pairs,
the total flow on an edge $e$ is defined as
$f(e) = \sum_{i=1}^k \sum_{p \in \cP_i: p \ni e} f_i(p)$.
The total flow for commodity $i$ is
$R_i  =  \sum_{p \in \cP_i} f_i(p)$
where $R_i$ is interpreted as the rate of commodity flow $i$.
In {\em directed} polymatroidal networks, the flow is constrained to satisfy
the following capacity constraints.
$$\sum_{e \in S} f({e}) \leq \rho_v^-(S) \quad  \forall v \, \forall S \subseteq \delta^-(v)
\quad \mbox{~and~} \quad \sum_{e \in S} f({e})  \leq  \rho_v^+(S) \quad \forall v \, \forall S \subseteq
\delta^+(v) $$

The constraints in {\em undirected} polymatroidal networks are:
$\sum_{e \in S} f({e})  \leq  \rho_v(S)  \quad \forall v \, \forall S \subseteq
\delta(v)$.

\medskip
A rate tuple $(R_1,...,R_k)$ is said to be {\em achievable} if
commodities $1,\ldots,k$ can be sent at rates $R_1,\ldots,R_k$
simultaneously between the corresponding source-sink pairs.  For a
given polymatroidal network and source-sink pairs the set of
achievable rate tuples is easily seen from the above constraints to be
a polyhedral set. We let $P(G,\cT)$ denote this rate region where $G$
is the network and $\cT$ is the set of given source-sink pairs.  In
the {\em maximum throughput multicommodity flow} problem the goal is
to maximize $\sum_{i=1}^k R_i$ over $P(G,\cT)$.  In the {\em maximum
  concurrent multicommodity flow} problem each source-sink pair has an
associated demand $D_i$ and the goal is to maximize $\lambda$ such
that the rate tuple $(\lambda D_1,...,\lambda D_k)$ is achievable,
that is the tuple belongs to $P(G,\cT)$. It is easy to see that both
these problems can be cast as linear programming problems. The
path-formulation results in an exponential (in $n$ the number of nodes
of $G$) number of variables and we also have an exponential number of
constraints due to the polymatroid constraints at each node. However,
one can use an edge-based formulation and solve the linear programs in
polynomial time via the ellipsoid method and polynomial-time
algorithms for submodular function minimization.

\medskip
\noindent
{\bf Network with symmetric demands:}
In directed polymatroidal networks we are primarily interested in {\em
  symmetric demands}: node $s_i$ intends to communicate with $t_i$ and
node $t_i$ intends to communicate with $s_i$ at the same {\em rate}.
Conceptually one can reduce this to the general setting by having two
commodities $(s_i,t_i)$ and $(t_i,s_i)$ for a pair $s_it_i$ and
adding a constraint that ensures their rates are equal. To be
technically consistent with previous work we do the following. We will
assume that we are given $k$ unordered source-sink pairs
$s_1t_1,\ldots,s_kt_k$. Now consider the $2k$ ordered pairs
$(s_1,t_1),\ldots,(s_k,t_k), (t_1,s_1),\ldots,(t_k,s_k)$.  We are
interested in achievable rate tuples of the form
$(R_1,\ldots,R_k,R'_1,\ldots,R'_k)$ where $R'_i = R_i$.  In the
maximum throughput setting we maximize $\sum_{i=1}^k (R_i+R'_i)$. Note
that even though the rates for $(s_i,t_i)$ and $(t_i,s_i)$ are the
same, the flow paths along which they route can be different.
In the maximum concurrent flow setting both $(s_i,t_i)$ and
$(t_i,s_i)$ have a common demand $D_i$ and we find the maximum
$\lambda$ such that rate tuple $(\lambda D_1,...,\lambda D_k,\lambda
D_1,...,\lambda D_k)$ is achievable for the pairs
$(s_1,t_1),\ldots,(s_k,t_k), (t_1,s_1),\ldots,(t_k,s_k)$.

\subsection{Cuts \label{sec:cuts}}
The multicommodity flow problems have natural dual cut problems
associated with them. Given a graph $G=(V,E)$ and a set of edges $F
\subseteq E$ we say that the ordered node pair $(s,t)$ is separated by
$F$ if there is no path from $s$ to $t$ in the graph $G[E\setminus
F]$. In directed graphs $F$ may separate $(s,t)$ but not $(t,s)$. In
undirected graphs we say that $F$ separates the unordered node pair
$st$ if $s$ and $t$ are in different connected components of
$G[E\setminus F]$. For certain problems, especially in the information
theoretic setting, it is of interest to consider restricted cuts
induced by vertex bi-partitions, that is cuts of the form $F =
\delta^+_G(S)$ ($F = \delta_G(S)$ in the undirected setting) for some
$S \subseteq V$. In this paper we mainly consider edge-cuts;
Section~\ref{sec:bipartition} discusses how vertex bi-partition cuts
can be obtained from edge-cuts for the sparsest cut problem in
undirected polymatroidal networks.

In the standard network model the cost of a cut
defined by a set of edges $F$ is simply $\sum_{e \in F} c(e)$ where
$c(e)$ is the cost of $e$ (capacity in the primal flow network) . In
polymatroid networks the cost of $F$ is defined in a more involved
fashion. Each edge $(u,v)$ in $F$ is assigned to either $u$ or $v$; we
say that an assignment of edges to nodes $g: F \rightarrow V$ is {\em
  valid} if it satisfies this restriction.  A valid assignment
partitions $F$ into sets $\{ g^{-1}(v) \mid v \in V\}$ where
$g^{-1}(v)$ (the preimage of $v$) is the set of edges in $F$ assigned
to $v$ by $g$.  For a given valid assignment $g$ of $F$ the cost of
the cut $\nu_g(F)$ is defined as
$$\nu_g(F) := \sum_v \left(\rho_v^-(\delta^-(v) \cap g^{-1}(v)) + \rho^+_v(\delta^+(v) \cap g^{-1}(v))\right).$$
In undirected graphs the cost for a given assignment is $\sum_v \rho_v(g^{-1}(v))$.

Given a set of edges $F$ we define its cost to be the minimum over
all possible valid assignments of $F$ to nodes, the expression for the
cost as above. We give a formal
definition below.

\begin{defn}
{\em Cost of edge cut:} Given a directed polymatroid network $G=(V,E)$
and a set of edges $F \subseteq E$, its cost denoted by $\nu(F)$ is
\begin{equation}
  \label{eq:dir-cost}
  \min_{g: F \rightarrow V, \text{~$g$ valid}} \sum_v \left(\rho_v^-(\delta^-(v) \cap g^{-1}(v)) + \rho^+_v(\delta^+(v) \cap g^{-1}(v))\right).
\end{equation}
In an undirected polymatroid network $\nu(F)$ is
\begin{equation}
  \label{eq:undir-cost}
  \min_{g: F \rightarrow V, \text{~$g$ valid}} \sum_v \rho_v(g^{-1}(v)).
\end{equation}
\end{defn}

\begin{lemma}
\label{lem:cut-subadditive}
 The cut cost function
is sub-additive,
 that is, $\nu(F \cup F') \le \nu(F) + \nu(F')$ for all $F,F' \subseteq E$.
\end{lemma}

Although not obvious, $\nu$ can be evaluated in polynomial time via an
algorithm to compute an $s$-$t$ maximum flow problem in a polymatroid
network. We do not, however, rely on it in this paper.

We now define the two cuts problems of interest.
\begin{defn}
  Given a collection of source-sink pairs $(s_1,t_1),\ldots,(s_k,t_k)$
  in $G=(V,E)$ and associated demand values $D_1,\ldots,D_k$, and a
  set of edges $F \subseteq E$ the demand separated by $F$, denoted by
  $D(F)$, is $\sum_{i:(s_i,t_i) \text{~separated by~} F} D_i$. $F$ is
  a {\em multicut} if all the given source-sink pairs are separated by $F$. The
  {\em sparsity} of $F$ is defined as $\frac{\nu(F)}{D(F)}$.
\end{defn}

The above definitions extend naturally to undirected graphs.  Given
the above definitions two natural optimization problems that arise are
the following. The first is to find a multicut of minimum cost for a
given collection of source-sink pairs. The second is to find a cut of
minimum sparsity. These problems are NP-Hard even in edge-capacitated
undirected graphs and have been extensively studied from an
approximation point of view
\cite{LeightonRao,GVY,LLR,AumannR98,ARV,AgarwalAC07}.

\begin{lemma}
  \label{lem:flow-less-than-cut}
  Given a multicommodity polymatroidal network instance, the value of
  the maximum throughput flow is at most the cost of a minimum
  multicut. The value of the maximum concurrent flow is at most the
  minimum sparsity.
\end{lemma}

A key question of interest is to quantify the relative gap between the
flow and cut values. These gaps are relatively well-understood in standard
networks and the main aim of this paper is to obtain results
for polymatroid networks.

\medskip
\noindent
{\bf Network with symmetric demands:}
For a directed network with symmetric demands the notion of a
``cut'' has to be defined appropriately. We say that a set of edges
$F$ separates a pair $s_it_i$ if it separates $(s_i,t_i)$ {\em
  or} $(t_i,s_i)$. With this notion of separation, the definitions of
multicut and sparsest cut extend naturally. A multicut is a set of
edges $F$ whose removal separates all the given pairs. Similarly for a
set of edges $F$ its sparsity is defined to $\nu(F)/D(F)$ where $D(F)$
is the total demand of pairs separated; note that if both $(s_i,t_i)$
and $(t_i,s_i)$ are separated by $F$ we count $D_i$ twice in $D(F)$.
This is to be consistent with the definition of flows given earlier.
Lemma~\ref{lem:flow-less-than-cut} extends to the symmetric demand
case with the definition of flows given for symmetric demands in
the previous section.

\section{Relaxations for Cuts}
\label{sec:cut-relaxations}
Lemma~\ref{lem:flow-less-than-cut} gives a way to lower bound the
value of multicuts and sparsest cuts via corresponding flow problems.
The flow problems can be cast as linear programs. The duals of these
linear programs can be directly interpreted as linear programming
relaxations for integer programming formulations for the cut problems.
Here we take the approach of writing the formulation with a convex
objective function and linear constraints; this simplifies and
clarifies the constraints and aids in the analysis. For one of the
cases we show the equivalence of the formulation with the dual of the
corresponding flow linear program. We first discuss continuous
extensions of submodular functions.

\subsection{Continuous extensions of submodular functions}
Given a submodular set function $\rho: 2^N\rightarrow \mathbb{R}$ on a
finite ground set $N$ it is useful to {\em extend} it to a function
$\rho':[0,1]^N \rightarrow \mathbb{R}$ defined over the cube in $|N|$
dimensions. That is we wish to assign a value for each $\bx \in
[0,1]^N$ such that $\rho'({\bf 1}_S) = \rho(S)$ for all $S\subseteq N$
where ${\bf 1}_S$ is the characteristic vector of the set $S$. For
minimizing submodular functions a natural goal is to find an extension
that is convex. We describe two extensions below.

\paragraph{Convex closure:} For a set function $\rho:2^N \rightarrow
\mathbb{R}$ (not necessarily submodular) its convex closure is a
function $\trho:[0,1]^N \rightarrow \mathbb{R}$ with
$\trho(\bx)$ defined as the optimum value of the following
linear program:
\begin{eqnarray*}
  \trho(\bx) & = & \min \sum_{S \subseteq N} \alpha_S \rho(S) \\
  \text{s.t.} && \\
  \sum_{S} \alpha_S & = & 1 \\
  \sum_{S: i \in S} \alpha_S & = & x_i \quad \quad \forall i \in N\\
  \alpha_S & \ge & 0 \quad \quad \forall S. \label{eq:conv_closure}
\end{eqnarray*}
The function $\trho$ is convex for any $\rho$. Moreover, when $\rho$
is submodular, for any given $\bx$, the linear program above can be
solved in polynomial time via submodular function minimization and
hence $\trho(\bx)$ can be computed in polynomial time (assuming a
value oracle for $\rho$). It is known and not difficult to show that
if $\rho$ is a polymatroid (monotone and $f(\emptyset) = 0$)
the value of the linear program does not change if we drop
the constraint that $\sum_S \alpha_S = 1$.

\paragraph{Lov\'{a}sz extension:} For a set function $\rho:2^N \rightarrow
\mathbb{R}$ (not necessarily submodular) its Lov\'{a}sz extension
\cite{Lovasz83} denoted by $\hrho:[0,1]^N \rightarrow \mathbb{R}$
is defined as follows:
$$ \hrho(\bx) = \int_0^1 \rho(\bx^\theta) d\theta$$
where $\bx^\theta = \{i \mid x_i \ge \theta \}$. This is not the standard
way the Lov\'{a}sz extension is stated but is entirely equivalent to it.
The standard definition is the
following. Given $\bx$ let $i_1,\ldots,i_n$ be a permutation of
$\{1,2,\ldots,n\}$ such that $x_{i_1} \ge x_{i_2} \ge \ldots \ge x_{i_n} \ge 0$.
For ease of
notation define $x_0=1$ and $x_{n+1} = 0$.
For $1 \le j \le n$ let $S_j = \{i_1,i_2,\ldots,i_j\}$.
Then
$$ \hrho(\bx) = (1-x_{i_1})\rho(\emptyset) + \sum_{j=1}^n (x_{i_j}-x_{i_{j+1}})\rho(S_j).$$
It is typical to assume that $\rho(\emptyset) = 0$ and omit the first
term in the right hand side of the preceding equation. Note that it is
easy to evaluate $\hrho(\bx)$ given a value oracle for $\rho$.

We state some well-known facts.

\begin{lemma}
  For a submodular set function $\rho$, $\trho(\bx) =
  \hrho(\bx)$ for any $\bx \in [0,1]^N$. Therefore the convex
  closure coincides with the Lov\'{a}sz extension and
  $\hrho(\cdot)$ is convex.
\end{lemma}

\begin{prop}
  \label{prop:monotone-sub}
  For a monotone submodular function $\rho$ and $\bx \le \bx'$
  (coordinate-wise), $\hrho(\bx) \le \hrho(\bx')$.
\end{prop}

The equivalence of $\trho$ and $\hrho$ also implies that an optimum
solution to the linear program defining $\trho(\bx)$ is obtained
by a solution $\bar{\alpha}$ where the support of $\bar{\alpha}$ is
a chain on $N$ (a laminar family whose tree representation is a path).
In fact we have the following. Given $\bx \in [0,1]^N$ consider
the ordering of the coordinates and the associated sets as in the
definition of the $\hrho(\bx)$. One can verify that
$\alpha_{S_j} = x_{i_{j}} - x_{i_{j-1}}$ for $1 \le j \le n$,
$\alpha_{\emptyset} = (1-x_{i_n})$, and  $\alpha_S = 0$ for all other
sets $S$ is an optimum solution to the linear program that defines
$\trho(\bx)$. We will use this fact later.

\subsection{Multicut}
\label{sec:multicut-relaxation}
We now consider the multicut problem. Recall that we wish to find a
subset $F \subseteq E$ such that $F$ separates all the given
source-sink pairs so as to minimize the cost $\nu(F)$. The only
difference between the polymatroid networks and standard networks is
in the definition of the cost. We first focus on expressing the
constraint that $F$ is a feasible set for separating the pairs.  For
each edge $e$ we have a variable $\ell(e) \in [0,1]$ in the relaxation
that represents whether $e$ is cut or not. For feasibility of the cut
we have the condition that for any path $p$ from $s_i$ to
$t_i$ (that is $p \in \cP_i$) at least one edge in $p$ is cut; in the
relaxation this corresponds to the constraint that $\sum_{e \in p}
\ell(e) \ge 1$. In other words $\dist_\ell(s_i,t_i) \ge 1$ where
$\dist_\ell(u,v)$ is the distance between $u$ and $v$ with edge
lengths given by $\ell(e)$ values.

We now consider the cost of the cut. Note that $\nu(F)$ is defined by
valid assignments of $F$ to the nodes, and submodular costs on the
nodes. In the relaxation we model this as follows. For an edge
$e=(u,v)$ we have variables $\ell(e,u)$ and $\ell(e,v)$ which decide
whether $e$ is assigned to $u$ or $v$.  We have a constraint
$\ell(e,u) + \ell(e,v) = \ell(e)$ to model the fact that if $e$ is cut
then it has to be assigned to either $u$ or $v$.  Now consider a node
$v$ and the edges in $\delta^+(v)$.  The variables $\ell(e,v), e \in
\delta^+(v)$ in the integer case give the set of edges $S \subseteq
\delta^+(v)$ that are assigned to $v$ and in that case we can use the
function $\rho^+_v(S)$ to model the cost.  However, in the fractional
setting the variables lie in the real interval $[0,1]$ and here we use
the extension approach to obtain a convex programming relaxation; we
can rewrite the convex program as an equivalent linear program via the
definition of $\trho$. Let $\bd^-_v$ be the vector consisting of the
variables $\ell(e,v)$, $e \in \delta^-(v)$ and similarly $\bd^+_v$
denote the vector of variables $\ell(e,v)$, $e \in \delta^+(v)$.  The
relaxation for the directed case is formally described in
Fig~\ref{fig:multicut-relaxations} in the box on the left.  For the
symmetric demands case the relaxation is similar but since
we need to separate either $(s_i,t_i)$ or $(t_i,s_i)$ the constraint
$\dist_\ell(s_i,t_i) \ge 1$ is replaced by the constraint
$\dist_\ell(s_i,t_i) + \dist_\ell(t_i,s_i) \ge 1$.

\begin{figure}[hb]
  \centering
  \begin{boxedminipage}[t]{0.45\linewidth}
    \begin{align}
\min \sum_v ( \hrho^-_v(\bd^-_v) & + \hrho^+_v(\bd^+_v)) && \label{eq:dir} \\
  \ell(e,u) + \ell(e,v) & =  \ell(e) \quad \quad  e = (u,v) \in E \nonumber \\
  \dist_\ell(s_i,t_i) & \ge  1 \quad \quad 1\le i\le k \nonumber \\
  \ell(e), \ell(e,u),\ell(e,v) & \ge  0 \quad \quad e = (u,v) \in E. \nonumber
\end{align}
\end{boxedminipage}
\hspace{1cm}
\begin{boxedminipage}[t]{0.45\linewidth}
\begin{align}
  \min \sum_v & \hrho_v(\bd_v) \label{eq:undir} \\
  \ell(e,u) + \ell(e,v) & =  \ell(e) \quad \quad  e = uv \in E \nonumber \\
  \dist_\ell(s_i,t_i) & \ge  1 \quad \quad 1\le i\le k \nonumber \\
  \ell(e), \ell(e,u),\ell(e,v) & \ge  0 \quad \quad e = uv \in E. \nonumber
\end{align}
\end{boxedminipage}
  \caption{Lov\'asz-extension based relaxations for multicut in directed and undirected polymatroidal networks}
  \label{fig:multicut-relaxations}
\end{figure}

For the undirected case we let $\bd_v$ denote the
vector of variables $\ell(e,v), e \in \delta(v)$ and the resulting
relaxation is shown on the right in Fig~\ref{fig:multicut-relaxations}.

One can replace $\hrho_v$ in the above convex programming relaxations
by $\trho_v$ the convex closure; further, one can use the
definition of $\trho_v$ via a linear program to convert the convex
program into an equivalent linear program. The resulting linear
program can be shown to be equivalent to the dual of the maximum
throughput flow problem. See Section~\ref{sec:formulations-equivalence}
for a formal proof.

\subsection{Sparsest cut}
\label{sec:sparsest-cut-relaxation}
Now we consider the sparsest cut problem. In the sparsest cut problem
we need to decide which pairs to disconnect and then ensure that we
pick edges whose removal separates the chosen pairs. Moreover we are
interested in the ratio of the cost of the cut to the demand
separated.  We follow the known formulation in the edge-capacitated
case with the main difference, again, being in the cost of the
cut. There is a variable $y_i$ which determines whether pair $i$ is
separated or not. We again have the edge variables $\ell(e),
\ell(e,u), \ell(e,v)$ to indicate whether $e=(u,v)$ is cut and whether
$e$'s cost is assigned to $u$ or $v$. If pair $i$ is to be separated
to the extent of $y_i$ we ensure that $\dist_\ell(s_i,t_i) \ge y_i$.
To express sparsity, which is defined as a ratio, we normalize the
demand separated to be $1$\footnote{We need to argue that
this leads to a valid relaxation for sparsest cut;
it follows from the fact that the polymatroid
functions in the objective function are normalized ($\rho_v(\emptyset) = 0$)
and in this case we have $\hrho_v(t \bx) = t \hrho_v(\bx)$ for any scalar
$t \in [0,1]$.}.  Fig~\ref{fig:sparsestcut-relaxations} has
a formal description on the left for the directed case.  For the
symmetric demands case we have essentially the same relaxation; the
constraint $\sum_i D_i \dist_\ell(s_i,t_i) = 1$ is replaced by the
constraint $\sum_i D_i (\dist_\ell(s_i,t_i) + \dist_\ell(t_i,s_i)) =
1$.

\begin{figure}[th]
  \centering
  \begin{boxedminipage}[t]{0.45\linewidth}
    \begin{align*}
    \min \sum_v \hrho^-_v(\bd^-_v) & + \hrho^+_v(\bd^+_v)) \\
  \ell(e,u) + \ell(e,v) & =  \ell(e) \quad \quad  e = (u,v) \in E \\
  \sum_{i=1}^k D_i \cdot \dist_\ell(s_i,t_i) & =  1 \\
  \ell(e), \ell(e,u),\ell(e,v) & \ge  0 \quad \quad e = (u,v) \in E.
    \end{align*}
\end{boxedminipage}
\hspace{1cm}
\begin{boxedminipage}[t]{0.45\linewidth}
\begin{align*}
  \min \sum_v & \hrho_v(\bd_v)  \\
  \ell(e,u) + \ell(e,v) & =  \ell(e) \quad \quad  e = uv \in E \\
  \sum_{i=1}^k D_i \cdot \dist_\ell(s_i,t_i)  & =  1 \\
  \ell(e), \ell(e,u),\ell(e,v) & \ge  0 \quad \quad e = (u,v) \in E.
\end{align*}
\end{boxedminipage}
  \caption{Relaxations for sparsest cut in directed and undirected polymatroidal networks}
  \label{fig:sparsestcut-relaxations}
\end{figure}

The relaxation for the undirected case is shown on the right
in Fig~\ref{fig:sparsestcut-relaxations} where
$\bd_v$ is the vector of variables $\ell(e,v), e \in \delta(v)$.

\section{Flow-Cut Gaps in Directed Polymatroidal Networks}
\label{sec:dir-gaps}
In this section we consider flow-cut gaps in directed polymatroidal
networks.  We show via a reduction that these gaps can be related to
corresponding gaps in directed edge-capacitated networks that have
been well-studied.  We note that this reduction is specific to
directed graphs and does not apply to undirected polymatroidal
networks. The embedding based approach for the undirected case that we
discuss in Section~\ref{sec:undir-gaps} is also applicable to directed
graphs.

The reduction is similar at a high-level for both gap questions of
interest and is based on the relaxations for the two cut problems that
we described in Section~\ref{sec:cut-relaxations}. We take a feasible
fractional solution for relaxation of the cut problem in question and
produce an instance of a cut problem in an edge-capacitated network
and a feasible fractional solution to the corresponding cut problem.
We also provide a correspondence between feasible integer solutions to
the edge-capacitated network instance and the original problem such
that the cost of the solution is preserved. These correspondences
allow us to translate known gap results for the edge-capacitated
networks to polymatroidal networks.

\subsection{Details of the reduction}
Let $G=(V,E)$ be a directed graph and let $\ell:E \rightarrow
\mathbb{R}_+$ be a length function on the edges. We let
$\dist_\ell(u,v)$ be the shortest path distance from $u$ to $v$ in $G$
with edge lengths $\ell$.  Moreover, for each edge $(u,v)$ let
$\ell(e,u)$ and $\ell(e,v)$ be two non-negative numbers such that
$\ell(e) = \ell(e,u)+\ell(e,v)$.  For a node $v$ let $\bd^+_v$ be the
vector of $\ell(e,v)$ values for all edges $e \in \delta^+(v)$ and
similarly $\bd^{-}_v$ is the vector of $\ell(e,v)$ values for edges in
$\delta^-(v)$. In the polymatroidal setting the cost induced by the
edge length variables is given by $\sum_{v\in V}(\hrho^-(\bd^-_v) +
\hrho^+(\bd^+_v))$. Note that for multicut we have that
$\dist_\ell(s_i,t_i) \ge 1$ for each demand pair $(s_i,t_i)$ while in
sparsest cut we are interested in the ratio of the cost to $\sum_{i}
D_i \cdot \dist_\ell(s_i,t_i)$.  We now describe the construction of a
graph $H=(V_H, E_H)$ where $V_H = V \uplus V'$ (that is the nodes of
$G$ are also in $H$) and an edge length function $\ell':E_H\rightarrow
\mathbb{R}_+$ such that $\dist_\ell(u,v) = \dist_{\ell'}(u,v)$ for all
$u,v \in V$; that is the distances between nodes in $V$ are the same
in $G$ and $H$. We also create an edge-cost (or capacity in the primal
sense) function $c: E_H \rightarrow \mathbb{R}_+$. The construction
will also establish the correspondence of cuts in $G$ and $H$ and their costs.

\begin{figure}[t]
\begin{center}
\includegraphics[width=6in]{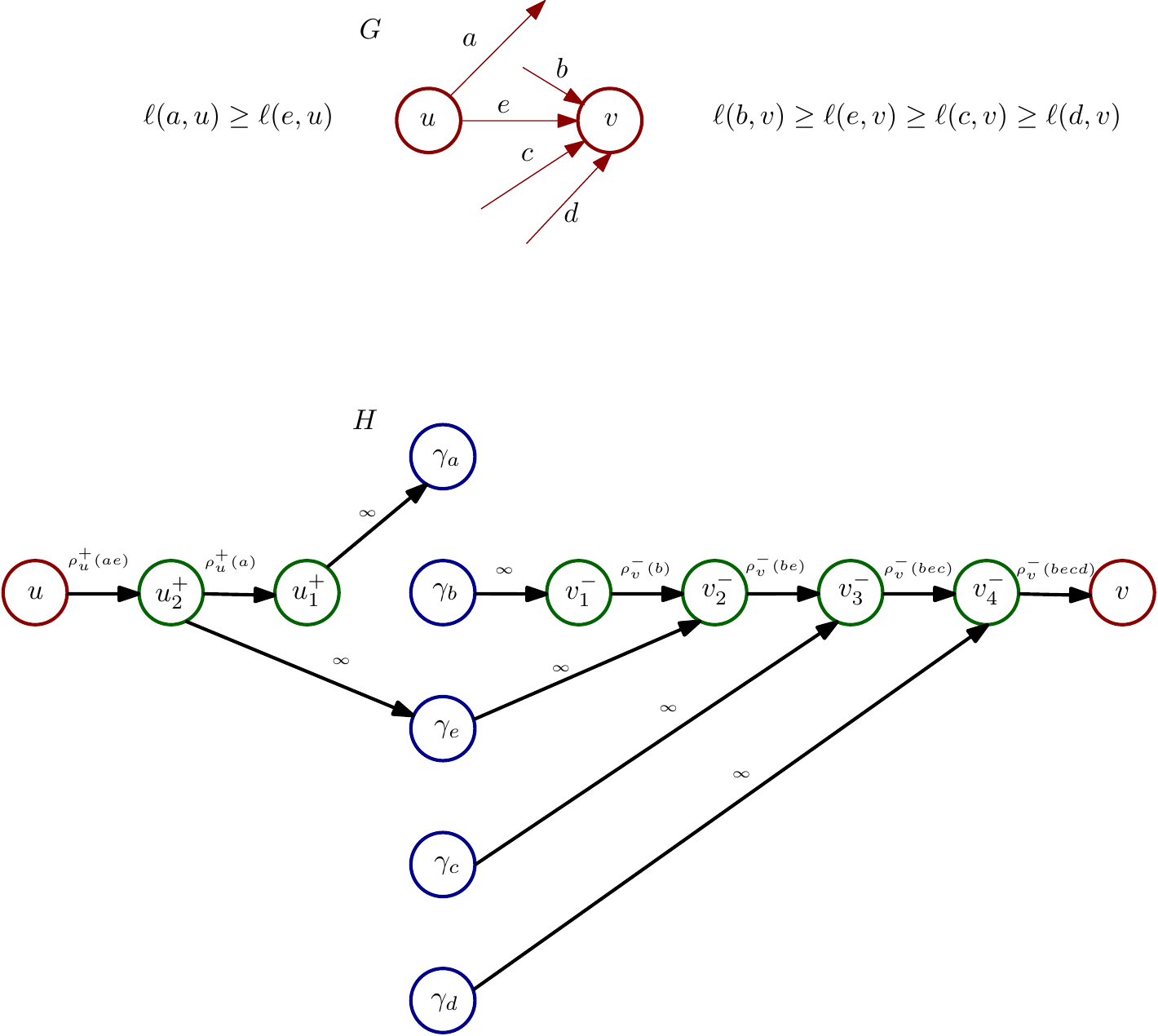}
\end{center}
\caption{Illustration of the reduction. Only $\delta_G^+(u)$ and
  $\delta_G^-(v)$ shown. The costs on edges in $H$ are shown but not
  their lengths.  The lengths of the infinite cost edges is $0$ and
  $\ell'(u^+_2,u^+_1) = \ell(a,u) - \ell(e,u)$ and $\ell'(v_3^-,v_4^-)
  = \ell(c,v)-\ell(d,v)$.} \label{fig:reduction}
\end{figure}

The graph $H=(V\uplus V',E_H)$ is constructed as follows. To aid the
reader we first describe the idea of the construction at a
high-level. Consider a node $v \in V$ and the in-coming edges
$\delta^-(v)$ and out-going edges $\delta^+(v)$.  In $H$ we have nodes
of $V$ and build an in-tree $T^-_v$ and an out-tree $T^+_v$ that are
rooted at $v$. The leaves of $T^-_v$ are the edges in $\delta^-(v)$
the leaves of $T^+_v$ are the edges in $\delta^+(v)$. Note that an
edge $(u,v)$ will thus participate in $T^+_u$ and $T^{-}_v$. Now for
the formal details. The nodes of $H$, denoted by $V_H$, consist of the
nodes $V$ of $G$ and additional nodes $V'$. $V'$ has two types of
nodes. First, for each edge $e \in E$ there is a node
$\gamma_e$. Second, for each node $v \in V$ we create two sets of
nodes $N^-(v)$ and $N^+(v)$ where $|N^-(v)| = n^-_v = |\delta_G^-(v)|$
and $|N^+(v)| = n^+_v = |\delta^+_G(v)|$; thus one node for each edge
in $\delta^-(v) \cup \delta^-(v)$; these will be the internal nodes of
the trees $T^{-}_v$ and $T^+_v$ respectively. For notational
convenience we refer to the $j$'th node in $N^-(v)$ as $v_j^-$ and
similarly $v^+_j$ for the $j$'th node in $N^+(v)$.

Now we describe the edge set $E_H$ of the graph $H$, the edge length
function $\ell':E_H\rightarrow \mathbb{R}_+$, and the cost function
$c:E_H\rightarrow \mathbb{R}_+$. The edge set is essentially
prescribed by specifying the trees $T^{-}_v$ and $T^+_{v}$ for each $v
\in V$. Consider the vector $\bd^-(v)$ of values $\ell(e,v)$ for $e
\in \delta_G^-(v)$. Recall the definition of the Lov\'{a}sz extension
$\hrho^-(\bd_v^-)$. We order the edges in $\delta^-(v)$ as
$e_1,e_2,\ldots,e_{n_v^-}$ where $\ell(e_j,v) \ge \ell(e_{j+1},v)$ for
$1\le j < n_v^-$ and then $\hrho^-(\bd_v^-) = \sum_j (\ell(e_j,v)-
\ell(e_{j+1},v)) \rho_v^-(S_j)$ where $S_j = \{e_1,\ldots,e_j\}$.  We
associate the node $v^-_j$ with the set $S_j$. The edge set of $T^-_v$
is defined as follows. For ease of notation we let $v^-_{n_v^-+1}$
represent the node $v$.  We create a directed path $v_1^- \rightarrow
v_2^- \rightarrow \ldots \rightarrow v_{n_v^-}^- \rightarrow
v_{n_v^-+1}^-= v$ with edge lengths $\ell'(v_1^-, v_2^-) = \ell(e_1,v)
- \ell(e_2,v), \ell'(v_2^-,v_3^-) = \ell(e_2,v)-\ell(e_3,v), \ldots,
\ell'(v_{n_v^-}^-,v) = \ell(e_{n_v^-},v)-0$.  The
costs of these edges are defined as follows: $c(v_j^-,v_{j+1}^-) =
\rho_v^-(S_j)$ for $1 \le j \le n_v^-$. For each $j$ we add the
edge $(\gamma_{e_j},v_j^-)$ with length $0$ and cost $\infty$ (for
computational purpose a sufficiently large number $M$ would do); this
connects the node $\gamma_{e_j}$ corresponding to the edge $e_j$ to
$v_j^-$ that corresponds to $S_j$.  See Fig~\ref{fig:reduction}.

The construction of $T^+_v$ is quite similar except that the edge
directions are reversed; assuming that the edges in $\delta^+(v)$ are
ordered such that $\ell(e_1,v) \ge \ell(e_2,v) \ge \ldots \ge
\ell(e_{n^+_v},v)$, we create a path $v \rightarrow v^+_{n^+_v}
\rightarrow \ldots v^+_2 \rightarrow v^+_1$ with edge lengths
$\ell(e_{n^+_v},v)-0, \ldots, \ell(e_{j},v) - \ell(e_{j+1},v), \ldots,
\ell(e_1,v)-\ell(e_2,v)$. The costs for the edges in this path are set
to $\rho^+_v(S_{n_v^+}),\ldots,\rho_v^+(S_1)$ where $S_j =
\{e_1,\ldots,e_j\}$. For each $j$ we add an edge $(v^+_j,\gamma_{e_j})$ with
length $0$ and cost $\infty$. This finishes the description of $H$.
We now describe various properties of the graph $H$. Several of these
properties are staright forward from the description of the construction
and we omit proofs of the easy claims.

The proposition below asserts the cost of the fractional solution
in the edge-capacitated network $H$ is the same as the cost of
the fractional solution in the polymatroidal network $G$.
\begin{prop}
  \label{prop:total-vol}
  $\sum_{e \in E_H} c(e) \cdot \ell'(e) = \sum_{v\in V}(\hrho^-(\bd^-_v) +
\hrho^+(\bd^+_v))$.
\end{prop}

\begin{prop}  \label{prop:dir-reduc-distance}
  For any edge $e \in \delta_G^-(v)$ the length of the unique path in
  $T_v^-$ from the node $\gamma_e$ to $v$ is equal to $\ell(e,v)$.
  Similarly for $e \in \delta_G^+(v)$, the length of the unique path
  in $T^+_v$ from the node $v$ to the node $\gamma_e$ is equal to $\ell(e,v)$.
\end{prop}

We now establish a correspondence between paths in $G$ and $H$ that
connect nodes in $V$. Let $e=(u,v)$ be an edge in $G$. We obtain a
canonical path $q(u,v)$ from $u$ to $v$ in $H$ as follows: concatenate
the unique path from $u$ to $\gamma_e$ in $T^+_v$ with the unique path
from $\gamma_e$ to $v$ in $T_v^-$. For any two nodes $s,t \in V$ let
$\cP_G(s,t)$ be the set of (simple) $s$-$t$ paths on $G$ and similarly
$\cP_H(s,t)$ be the paths in $H$. We create a map $g: \cP_G(s,t)
\rightarrow \cP_H(s,t)$ as follows. Consider a path $p \in
\cP_G(s,t)$; we obtain a path $p' \in \cP_H(s,t)$ corresponding to $p$
as follows. We replace each edge $(u,v) \in p$ by the canonical path
$q(u,v)$.

\begin{lemma}
  \label{lem:dir-reduc-distance}
  The map $g$ is a bijection. Moreover, for any two nodes $u,v \in V$,
  $\dist_{\ell'}(u,v) = \dist_\ell(u,v)$.
\end{lemma}

Now we establish a correspondence between cuts in $G$ and $H$.  For a
given set of edges $F \subseteq E$ let $\sep_G(F)$ be set of node pairs in
$V \times V$ separated by $F$ in the graph $G$. Similarly for a set of
edges $F' \subseteq E_H$ let $\sep_H(F')$ be the set of node pairs in $V
\times V$ separated by $F'$ in the graph $H$. We say that a set of
edges $F$ is minimal with respect to separating node pairs if there is
no proper subset of $F$ that separates the same node pairs as $F$.

\begin{prop}
  \label{prop:minimal-sep}
  Let $F' \subseteq E_H$ be minimal with respect to separating node
  pairs in $V \times V$ and of finite cost. Then for any $v \in V$,
  $F'$ contains at most one edge from $T_v^-$ and at most one edge
  from $T_v^+$.
\end{prop}
\begin{proof}
  Consider a node $v$ and edge-sets $F' \cap T_v^-$ and $F' \cap
  T_v^+$.  For an edge $e \in E$ there is a node $\gamma_e \in V_H$
  and there is exactly one edge coming into $\gamma_e$ and exactly one
  edge going out of $\gamma_e$ and both are of infinite cost.
  Therefore, if $F'$ is of finite cost, $F' \cap T_v^-$ consists of
  some edges in the path $v_1^- \rightarrow v_2^- \ldots \rightarrow
  v_{n_v^-}^- \rightarrow v$ contained in $T_v^-$. Since the only way
  to reach $v$ is through $T_v^-$ it follows that if $F'$ contains an
  edge $(v_j^-,v^-_{j+1})$ then it is redundant to remove an edge
  $(v_{i}^-,v_{i+1}^-)$ for $i < j$. Thus minimality of $F'$ implies
  $F'$ contains exactly one edge from $T_v^-$. The reasoning for
  $T_v^+$ is similar.
\end{proof}

\begin{lemma}
  \label{lem:dir-reduc-cost}
  Let $F' \subseteq E_H$ be minimal with respect to separating node
  pairs in $V \times V$ and of finite cost. There exists a set of
  edges $F \subseteq E$ such that $\sep_G(F) \supseteq
  \sep_H(F')$ and $\nu(F) \le c(F')$.
\end{lemma}
\begin{proof}
  Given a minimal $F'$ we obtain a set of edges $F \subseteq E$ as follows.
  From the proof of Proposition~\ref{prop:minimal-sep} we see that for
  any node $v$, $F'$ contains at most one edge from $T_v^-$ and in
  particular if it contains an edge then it is an edge
  $(v_j^-,v^-_{j+1})$ for some $1 \le j \le n_v^-$ (for simplicity we
  identify $v$ with $v^{-}_{n_v^-+1}$). Suppose there is such an edge
  $e'=(v_j^-,v^-_{j+1})$ in $F'$. Note that $e'$ corresponds to the
  set $S_j = \{e_1,\ldots,e_j\}$ of edges in $\delta_G^-(v)$ ordered
  in increasing order by $\ell(e,v)$ values. We add $S_j$ to $F$
  and assign these edges to $v$ in upper bounding $\nu(F)$: by construction
  $c(e') = \rho_v^-(S_j)$. We do a similar procedure if $e' \in F \cap T_v^+$.
  It follows that the edge set $F$ that we construct satisfies the property
  that $\nu(F) \le c(F')$.

  We now show that $\sep_G(F) \supseteq  \sep_H(F')$
  Consider a pair $(s,t)$ such that $s$ is separated from $t$ by $F'$ in $H$.
  Suppose $(s,t)$ is not separated by $F$ in $G$. Let $p$ be an $s$-$t$ path
  that remains in $G\setminus F$. From Proposition~\ref{prop:dir-reduc-distance}
  there is a unique path $g(p) \in \cP_H(s,t)$. For every edge $e=(u,v) \in p$
  consider the canonical path $q(u,v)$ in $H$. Since $e$ is not in $F$
  it implies that $u$ can reach $\gamma_e$ in $H\setminus F'$ and that
  $\gamma_e$ can reach $v$ in $H \setminus F'$. This means that $q(u,v)$
  exists in $H\setminus F'$. This would imply that $g(p)$ exists in
  $H\setminus F'$ contradicting that assumption that $(s,t)$ is separated
  by $F'$.
\end{proof}

We summarize the properties of the reduction. We assume that we have a
polymatroidal network $G=(V,E)$ with $k$ demand pairs
$(s_i,t_i),\ldots,(s_k,t_k)$ with associated demand values
$D_1,\ldots,D_k$. For all the cut problems of interest, the
relaxations in Section~\ref{sec:cut-relaxations} produce a length
function $\ell:E\rightarrow \mathbb{R}_+$ and for each $e=(u,v)$
associated non-negative values $\ell(e,u)$ and $\ell(e,v)$ such that
$\ell(e) = \ell(e,u) + \ell(e,v)$. As before we use $\bd_v^-$ and
$\bd_v^+$ to denote the vector of $\ell(e,v)$ values for the incoming
and outgoing edges at $v$. The reduction produces an edge-capacitated
network $H=(V_H,E_H)$ with the following properties:
\begin{itemize}
\item each node of $V$ is a node in $V_H$
\item for all $u,v \in V$, $\dist_\ell(u,v) = \dist_{\ell'}(u,v)$
\item $\sum_{e \in E_H} c(e) \ell'(e) = \sum_{v \in V}(\hrho_v^-(\bd_v^-) + \hrho_v^+(\bd_v^+))$
\item for any set of edges $F'\subseteq E_H$ there is a corresponding
  set $F \subseteq E$ such that $\sep_G(F) \supseteq \sep_H(F')$
  and $\nu(F) \le c(F')$.
\end{itemize}

We also note that the reduction can be carried out in polynomial time. Moreover,
given a set $F' \subseteq E_H$ a set $F \subseteq E$ that satisfies the
last property in the list above can be found in polynomial time.

We build on the reduction to obtain flow-cut gap results, all of which
are based on using the relaxations from
Section~\ref{sec:cut-relaxations} which are dual to the corresponding
flow problems.  We argue via the reduction and known results on
edge-capacitated networks that there exist integral cuts within some
factor $\alpha$ of the fractional solution.

\subsection{Multicut} We consider the multicut problem for arbitrary
demand pairs as well as symmetric demands.  The relaxation satisfies
the constraint that $\dist_\ell(s_i,t_i) \ge 1$ for each demand pair
$(s_i,t_i)$.  The reduction from the preceding section produces a
graph $H=(V_H,E_H)$ and a fractional solution $\ell':E_H \rightarrow
\mathbb{R}_+$ such that $\dist_{\ell'}(s_i,t_i) \ge 1$.
We note that $\ell'$ is a feasible solution for the standard distance
based relaxation for multicut in edge-capacitated networks which is
the dual for the maximum throughput multicommodity flow problem.
The integrality gap of this relaxation has been studied and several
results are known. Let $\beta = \sum_{e \in E_H} c(e) \ell'(e)$ be
the fractional solution value. Then one can obtain an integral
multicut $F'$ with cost $c(F')$ that can be bounded in terms of
$\beta$. We summarize the known results.

\begin{itemize}
\item Cheriyan, Karloff and Rabani \cite{CheriyanKR05} showed that
  there exists an $F'$ such that $c(F') \le O(1)\cdot \beta^3$; this
  was improved by Gupta~\cite{Gupta03} to show the existence of a
  multicut $F'$ such that $c(F') \le O(1) \cdot \beta^2$.  These
  results hold under the assumption that $c(e) \ge 1$ for all $e$.
\item Agrawal, Alon and Charikar \cite{AgarwalAC07} improving the
  results in \cite{CheriyanKR05,Gupta03} showed the existence of
  a cut $F'$ such that $c(F') = \tilde{O}(n^{11/23}) \cdot \beta$. Here $n$
  is the number of nodes in the graph.
\item Saks, Samorodnitsky and Zosin \cite{SaksSZ04} showed that there exist
  instances on which every integral multicut has a value $\Omega(k) \cdot \beta$.
\item Chuzhoy and Khanna \cite{ChuzhoyK07} showed that there exist instances
  on which every multicut has a value $\tilde{\Omega}(n^{1/7}) \cdot \beta$. Further, they
  showed that the multicut problem is hard to approximate to within
  a factor of $\Omega(2^{\log^{1-\epsilon}n})$ unless $NP \subseteq ZPP$.
\end{itemize}

Since polymatroidal networks generalize edge-capacitated networks it follows
that all the lower bounds in the above hold for the polymatroidal network case
as well. The reduction also allows us to obtain upper bound for polymatroidal
networks. We have to careful when using bounds that depend on  the
number of nodes in the graph. The reduction takes $G$ with $n$ nodes
and $m$ edges and produces an edge-capacitated graph $H$ with
$n+2m$ nodes. In the worst case $H$ has $\Omega(n^2)$ nodes.
We thus obtain the following theorem.

\begin{theorem}
  In a directed polymatroidal network $G$ on $n$ nodes, for any given
  multicommodity flow instance with $k$ pairs, if $\beta$ is the maximum
  throughput multicommodity flow then:
  \begin{itemize}
  \item There is a feasible multicut $F'$ such that $\nu(F') \le O(1)
    \cdot \beta^2$ assuming that $\rho^+_v$ and $\rho_v^-$
    are integer valued for all $v \in V$.
   \item There is a feasible multicut $F'$ such that
     $\nu(F') \le \tilde{O}(n^{22/23}) \cdot \beta$.
  \end{itemize}
  Moreover, there exist polynomial-time algorithms to find
  multicuts guaranteed as above.
\end{theorem}

\paragraph{Symmetric demands:} We now consider the symmetric demand
case when a multicut corresponds to separating $(s_i,t_i)$ or
$(t_i,s_i)$ for a given demand pair $s_it_i$. The relaxation for this
has a constraint that $\dist_\ell(s_i,t_i) + \dist_\ell(t_i,s_i) \ge
1$.  In contrast to the strong negative results for the general
multicut problem, poly-logarithmic upper bounds on flow-cut gaps are known for
symmetric demands in standard networks. In particular Klein
\etal \cite{KleinPRT97} show that if $\beta$ is the cost of a fractional
solution then there exists an integral multicut of cost $O(\log^2 k)
\cdot \beta$.  Even \etal \cite{ENRS} showed the existence of a multicut of
cost $O(\log n \log \log n) \cdot \beta$. Note that these bounds are
incomparable in that depending on the relationship between $k$ and $n$
one is better than the other. It is also known that there exist instances
on which the gap is at least $\Omega(\log n)$. Via the reduction we obtain
the following.

\begin{theorem}
  In a directed polymatroidal network $G$ on $n$ nodes, for any given
  multicommodity flow instance with symmetric demands on $k$ pairs,
  the minimum multicut is $O(\min\{\log^2 k, \log n \log \log n\}) \cdot \beta$
  where $\beta$ is maximum throughput multicommodity flow for the symmetric
  demands.
\end{theorem}

\begin{remark}
  The flow-cut gap in polymatroidal networks for multiterminal
  flows\footnote{In multiterminal flows we have a set of $k$ terminals
    $\{s_1,s_2,\ldots,s_k\}$ and flow can be sent between any pair of
    terminals; the goal is to maximize the total flow. The
    corresponding cut is referred to as multiterminal cut or multiway
    cut in which the goal is to remove a minimum-cost set of edges to
    disconnect every (ordered) pair of terminals.}  can be shown to be
  $2$ via the reduction and the result of Naor and Zosin \cite{NaorZ01}.
\end{remark}

\subsection{Sparsest cut}
Now we consider the sparsest cut problem where the goal is to find a
set of edges $F$ to minimize $\nu(F)/D(F)$ where $D(F)$ is the total
demand of the pairs separated by $F$. The relaxation corresponds to
finding edge length variables $\ell$ to minimize the fractional cost
subject to the constraint that $\sum_i D_i \cdot \dist_\ell(s_i,t_i) =
1$.  Via the reduction we produce an edge-capacitated network $H$ such
that $\sum_i D_i \cdot \dist_{\ell'}(s_i,t_i) = 1$ and with the
fractional cost preserved. In edge-capacitated networks there is a
generic strategy that translates the flow-cut gap for multicut into a
flow-cut gap for sparsest cut at an additional loss of an $O(\log
\sum_i D_i)$ factor due to Kahale \cite{Kahale} (see also
\cite{Shmoys-survey}); this has been refined via a more intricate
analysis in \cite{PlotkinT95} to lose only an $O(\log k)$ factor
although one needs to apply it carefully. In \cite{AgarwalAC07} a
simple reduction that loses an $O(\log n)$ factor is given (this builds
on \cite{Kahale}). For
directed graphs the known-gaps for sparsest cut are essentially
based on using the corresponding gap for multicut and translating via
the above mentioned schemes.  We thus obtain the following results.

\begin{theorem}
  In a directed polymatroidal network $G$ on $n$ nodes, for any given
  multicommodity flow instance with $k$ pairs, if $\beta$ is the value
  of the maximum concurrent flow then there is a cut of sparsity at
  most $\tilde{O}(n^{22/23}) \cdot \beta$.
\end{theorem}

\begin{theorem}
  In a directed polymatroidal network $G$ on $n$ nodes, for any given
  multicommodity flow instance with symmetric demands on $k$ pairs,
  there is a cut of sparsity $O(\min\{\log^3 k, \log^2 n \log \log n\}) \cdot \beta$   where $\beta$ is maximum concurrent flow.
\end{theorem}

\section{Flow-Cut Gaps in Undirected Polymatroidal Networks}
\label{sec:undir-gaps}
In this section we consider flow-cut gaps in undirected polymatroidal
networks.  As we already noted, node-capacitated flows are a special
case of polymatroidal flows. We show that line embeddings with low
average distortion introduced by Matousek and Rabinovich
\cite{MatousekR01} (and further studied in \cite{Rabinovich03}) are
useful for bounding the gap between the maximum concurrent flow and
sparsest cut; we are inspired to make this connection from
\cite{FeigeHL06} who considered node-capacitated flows. For multicut
we show that the region growing technique from \cite{LeightonRao} that
was used in \cite{GVY} for edge-capacitated multicut can be adapted to
the polymatroidal setting.  These techniques are also applicable to
directed graphs --- we defer a more detailed discussion.

\subsection{Maximum Concurrent Flow and Sparsest Cut}
We start with the definition of line embeddings and average distortion.

Let $(V,d)$ be a finite metric space. A map $g : V \rightarrow
\mathbb{R}$ is an embedding of $V$ into a line; it is a {\em
  contraction} (also called $1$-Lipschitz) if for all $u,v \in V$,
$$ |g(u)-g(v)| \le d(u,v).$$
Given a demand function $w: V \times V \rightarrow \mathbb{R}_+$
and a contraction $g: V \rightarrow \mathbb{R}$, its {\em average distortion}
with respect to $w$ is defined as
$$ \avgd_w(g) = \frac{\sum_{u,v\in V} w(u,v) \cdot d(u,v)}{ \sum_{u,v
    \in V} w(u,v) \cdot |g(u)-g(v)|}
$$

The following theorem is implicit in \cite{Bourgain}; see \cite{FeigeHL06}
for a sketch.

\begin{theorem}[Bourgain \cite{Bourgain}]
\label{thm:Bourgain-line}
For every $n$-point metric space $(V,d)$ and
every weight function  $w: V \times V \rightarrow \mathbb{R}_+$ there
is a polynomial-time computable contraction $g : V \rightarrow \mathbb{R}$ such
that $\avgd_w(g) = O(\log n)$. Moreover, if the support of $w$ is $k$
there is a map $g$ such that $\avgd_w(g) = O(\log k)$.
\end{theorem}

Using the above we prove the following.

\begin{theorem}
\label{thm:undir-sparse-gap}
In undirected polymatroidal networks, for any given multicommodity
flow instance with $k$ pairs, the ratio between the value of the
sparsest cut and the value of the maximum concurrent flow is $O(\log
k)$. Moreover, there is an efficient algorithm to compute an $O(\log
k)$ approximation to the sparsest cut problem.
\end{theorem}

Recall the relaxation for the sparsest cut from
Section~\ref{sec:sparsest-cut-relaxation} and the associated notation.
To prove the theorem we consider an optimum solution to the relaxation
and show the existence of a cut whose sparsity is $O(\log k)$ times
the value of the relaxation. Let $(V,d)$ be the metric induced on $V$
by shortest path distances in the graph with edge lengths given by
$\ell:E \rightarrow \mathbb{R}_+$ from the optimum fractional
solution. Let $g:V \rightarrow \mathbb{R}$ be line embedding
guaranteed by Theorem~\ref{thm:Bourgain-line} with respect to $d$ and
the weight function given by the demands $D_i$; that is $w(s_i,t_i) =
D_i$ for a demand pair and is $0$ for any pair of nodes that do not
correspond to a demand. Without loss of generality we can assume that
$g$ maps $V$ to the interval $[0,\beta]$ for some $\beta > 0$. For
$\theta \in (0,\beta)$ let $S_\theta = \{ u \mid g(u) \le \theta
\}$. We show that there is a $\theta$ such that $\delta(S_\theta)$ is
an approximately good sparse cut. Let $D(\delta(S_\theta))$ be the
total demand of pairs separated by $S_\theta$, that is
$D(\delta(S_\theta)) = \sum_{i: S_\theta \text{~separates~} s_it_i} D_i$.

\begin{lemma}
$$\int_0^\beta D(\delta(S_\theta)) d \theta= \Omega\left (\frac{1}{\log k}\right ).$$
\end{lemma}
\begin{proof}
From the definition of $D(\delta(S_\theta))$,
$$\int_0^\beta D(\delta(S_\theta)) d \theta = \int_0^\beta (\sum_{i: S_\theta \text{~separates~} s_it_i} D_i)  d \theta = \sum_{i=1}^k D_i \cdot \int_0^\beta {\bf 1}_{S_\theta \text{~separates~} s_it_i} d\theta=  \sum_{i=1}^k D_i \cdot |g(s_i)-g(t_i)|.$$
From the properties of $g$,
$$\frac{\sum_i D_i \cdot d(s_i,t_i)}{\sum_{i} D_i \cdot |g(s_i)-g(t_i)|} \le O(\log k).$$
We have the constraint $\sum_i D_i \cdot d(s_i,t_i) = 1$ from the LP
relaxation; this combined with the above inequality proves the lemma.
\end{proof}

The main insight in the proof is the following lemma. A version of
the lemma also holds for directed graphs that we address
in a remark following the proof.

\begin{lemma}
\label{lem:line-key}
$$\int_0^\beta \nu(\delta(S_\theta)) d \theta \le 2 \sum_u \hrho_u(\bd_u).$$
\end{lemma}

\begin{proof}
  Consider an edge $uv \in \delta(S_\theta)$ and for simplicity
  assume $g(u) < g(v)$.  The length of $e$ in the embedding
  is $\ell'(e) = |g(v)-g(u)| \le \ell(e)$. The edge $(u,v) \in
  \delta(S_\theta)$ iff $\theta$ is in the interval
  $[g(u),g(v)]$. Note that the cost $\nu(\delta(S_\theta))$ is in
  general a complicated function to evaluate. We upper bound
  $\nu(\delta(S_\theta))$ by giving an explicit way to assign $e=uv$
  to either $u$ or $v$ as follows. Recall that in the relaxation $\ell(e) =
  \ell(e,u) + \ell(e,v)$ where $\ell(e,u)$ and $\ell(e,v)$ are the
  contributions of $u$ and $v$ to $e$. Let $r =
  \frac{\ell(e,u)}{\ell(e)}$ and let $\ell'(e,u) = r \ell'(e)$ and
  $\ell'(e,v) = (1-r)\ell'(e)$.  We partition the interval
  $[g(u),g(v)]$ into $[g(u),g(u)+\ell'(e,u))$ and
  $[g(u)+\ell'(e,u),g(v)]$; if $\theta$ lies in the former interval we
  assign $e$ to $u$, otherwise we assign $e$ to $v$.  This assignment
  procedures describes a way to upper bound $\nu(\delta(S_\theta))$
  for each $\theta$. Now we consider the quantity $\int_0^\beta
  \nu(\delta(S_\theta)) d\theta$ and upper bound it as follows.

  Consider a node $u$ and let $L_u = \{uv \in \delta(u) \mid g(v) <
  g(u) \}$ be the set of edges $uv$ that go from $u$ to the left of
  $u$ in the embedding $g$. Similarly $R_u = \{uv \in \delta(u) \mid
  g(v) \ge g(u) \}$. Note that $L_u$ and $R_u$ partition $\delta(u)$.
  Let $\bd'_u$ be the vector of dimension $|\delta(u)|$ consisting of
  the values $\ell'(e,u)$ for $e \in \delta(u)$.  We obtain $\bd^L_u$
  from $\bd'_u$ by setting the values for $e \in R_u$ to $0$ and
  similarly $\bd^R_u$ from $\bd'_u$ by setting the values for $e \in
  L_u$ to $0$. Since $0 \le \ell'(e,u) \le \ell(e,u)$ for each $e \in
  \delta(u)$ we see that $\bd'_u \le \bd_u$ and (component wise) and
  hence $\bd^L_u \le \bd_u$ and $\bd^R_u \le \bd_u$.  Since $\rho_u$
  is monotone we have that $\hrho_u(\bd^L_u) \le \hrho_u(\bd_u)$ and
  $\hrho_u(\bd^R_u) \le \hrho_u(\bd_u)$ (see
  Proposition~\ref{prop:monotone-sub}).

  We claim that
$$\int_0^\beta \nu(\delta(S_\theta)) d\theta \le \sum_{u \in V} (\hrho_u(\bd^L_u) + \hrho_u(\bd^R_u)),$$
which would prove the lemma.

To see the claim consider some fixed $\theta$ and
$\nu(\delta(S_\theta))$. Fix a node $u$ and consider the edges in
$\delta(u) \cap S_\theta$ assigned to $u$ by the procedure we
described above; call this set $A_{\theta,u}$. First assume that
$\theta < g(u)$. Then the edges assigned to $u$ by the procedure,
denoted by $A_{\theta,u} = \{ e \in L_u \mid \theta > g(u) -
\ell'(e,u) \}$. Similarly, if $\theta > g(u)$, $A_{\theta,u} = \{ e
\in L_u \mid \theta < g(u) + \ell'(e,u) \}$. From these definitions we have
$$\int_0^\beta \nu(\delta(S_\theta)) d\theta \le
\sum_{u \in V} \int_0^\beta \rho_u(A_{\theta,u})d\theta.$$
For a fixed node $u$,
\begin{eqnarray*}
\int_0^\beta \rho_u(A_{\theta,u})d\theta & = &   \int_0^{g(u)} \rho_u(A_{\theta,u})d\theta + \int_{g(u)}^\beta \rho_u(A_{\theta,u})d\theta
\end{eqnarray*}
Let $L_u = \{e_1,e_2,\ldots, e_h\}$ where $0 \le \ell'(e_1,u) \le
\ell'(e_2,u) \le \ldots \le \ell'(e_h,u)$. Then
$$\int_0^{g(u)} \rho_u(A_{\theta,u})d\theta = \sum_{j=1}^h (\ell'(e_j,u)-\ell'(e_{j-1},u))\rho(\{e_1,e_2,\ldots,e_j\})$$
The right hand side of the above, is by construction and the
definition of the Lov\'{a}sz extension, equal to $\hrho_u(\bd^L_u)$.
Similarly, $\int_{g(u)}^\beta \rho_u(A_{\theta,u})d\theta =
\hrho_u(\bd^R_u)$.
\end{proof}

\begin{remark}
  An examination of the proof of the above lemma explains the factor
  of $2$ on the right hand side; the edges in $\delta(v)$ can be both
  to the left and right of $v$ in the line embedding and each side
  contributes $\hrho_u(\bd_v)$ to the cost. This is related to the
  technical issue about undirected polymatroid networks where the flow
  through $v$ takes up capacity on two edges incident to $v$.  For
  directed graphs one can prove a statement of the form below where
  $\delta^+(S_\theta)$ is set of edges leaving $S_\theta$. Notice that
  there is no factor of $2$ since one treats the incoming and outgoing
  edges separately.
  $$\int_0^\beta \nu(\delta^+(S_\theta)) d \theta \le  \sum_u(\hrho_u^-(\bd^-_u) + \hrho_u^+(\bd^+_u)).$$
  The above statement gives an embedding proof of the maxflow-mincut
  theorem for single-commodity directed polymatroidal networks and has
  other applications.
\end{remark}

\medskip
\noindent
We now finish the proof of Theorem~\ref{thm:undir-sparse-gap} via the
preceding two lemmas.
\begin{eqnarray*}
\min_{\theta \in (0,\beta)} \frac{\nu(\delta(S_\theta))}{D(\delta(S_\theta))} & \le & \frac{\int_0^\beta \nu(\delta(S_\theta)) d \theta}{\int_0^\beta D(\delta(S_\theta)) d \theta} \\
     & \le & 2 \sum_u \hrho_u(\bd_u) \cdot O(\log k) = O(\log k) \sum_u \hrho_u(\bd_u).
\end{eqnarray*}
The above shows that the sparsity of $S_\theta$ for some $\theta$ is
at most $O(\log k)$ times $ \sum_u \hrho_u(\bd_u)$ which is the value of
the relaxation. Given a line embedding $g$ there are only $n-1$
distinct cuts of interest and one can try all of them to find the one
with the smallest sparsity. The efficiency of the algorithm therefore
rests on the efficiency of the algorithm to solve the fractional
relaxation, and the algorithm to find a line embedding guaranteed
by Theorem~\ref{thm:Bourgain-line}; both have polynomial time algorithms
and thus one can find an $O(\log k)$ approximation to the sparsest cut
in polynomial time.

\begin{remark}
  Node-weighted flows and cuts/separators can be cast as special cases
  of flows and cuts in polymatroid networks. Our algorithm produces
  edge-cuts from line embeddings in a simple way even for
  node-weighted problems --- the $\nu$ cost of the edge-cut
  automatically translates into an appropriate node-weighted cut.  In
  contrast, the algorithm in \cite{FeigeHL06} has to solve several
  instances of $s$-$t$ separator problems in auxiliary graphs obtained
  from the line embedding.
\end{remark}

\subsubsection{Sparsest Bi-partition Cut}
\label{sec:bipartition}
We worked with general edge cuts so far, but for certain applications,
it is necessary to work with a special type of edge cut called a
bi-partition cut. In an undirected polymatroidal network, an edge-cut
$F$ is said to be a {\em bi-partition cut} if there exists a set $S
\subseteq V$ such that $F = \delta_G(S)$.  In the case of
edge-capacitated undirected networks, it is well known that for any
multicommodity flow instance, there always exists a sparsest cut that
is a bi-partition cut. This does not hold for
polymatroidal networks, however, a factor $2$ gap can indeed be shown between
the sparsest cut and the sparsest cut restricted to bi-partition
cuts; moreover this factor is tight.

\bthm \label{thm:bi-partition} Given any edge cut $F$ for a
multicommodity flow instance in an undirected polymatroidal network
$G=(V,E)$, there exists a bi-partition cut $\delta_G(S)$ whose sparsity is
atmost $2$ times the sparsity of $F$. Furthermore this factor
is tight. \ethm

The proof of the above theorem can be found in
Section~\ref{app:bi-partition} of the appendix.
Theorem~\ref{thm:undir-sparse-gap} and
Theorem.~\ref{thm:bi-partition} together imply a logarithmic gap
between maximum concurrent flow and sparsest bi-partition cut. This is
formally stated in the following corollary.

\bcor In undirected
polymatroidal networks, for any given multicommodity flow instance
with $k$ pairs, the ratio between the value of the {\em sparsest
  bi-partition cut} and the value of the maximum concurrent flow is
$O(\log k)$.  \ecor


\subsection{Maximum Throughput Flow and Multicut}
We prove the following theorem in this section.

\begin{theorem}
\label{thm:undir-multicut-gap}
In undirected polymatroidal networks, for any given multicommodity
flow instance with $k$ pairs, the ratio between the value of the
minimum multicut and the value of the maximum throughput flow is $O(\log
k)$. Moreover, there is an efficient algorithm to compute an $O(\log
k)$ approximation to the minimum multicut problem.
\end{theorem}

We recall the relaxation for the minimum mulitcut problem
from Section~\ref{sec:multicut-relaxation}. Consider an optimum
solution to the relaxation given by edge lengths $\ell(e), e \in E$ and
the partition of $\ell(e)$ for each $e=uv$ between $u$ and $v$ given
by the variables $\ell(e,u)$ and $\ell(e,v)$. We will show that
there exists a multicut $F \subseteq E$ for the given pairs such
that $\nu(F) = O(\log k) (\sum_v \hrho_v(\bd_v))$.

By slightly generalizing the proof of Lemma~\ref{lem:line-key} we obtain
the following.

\begin{lemma}
  \label{lem:line-refined}
  Let $g:V \rightarrow [0,\beta]$ be a contraction, let $0 \leq a_0 \leq a < b \leq b_0 \leq \beta$ and $S_\theta = \{u \mid g(u) <  \theta\}$. Suppose for every edge $e=uv \in \cup_{\theta \in [a,b]} \delta(S_\theta)$, $g(u)$ and $g(v)$ are both in $[a_0,b_0]$. Then,
  $$\int_{a}^b \nu(\delta(S_\theta)) d\theta \le 2 \sum_{v: g(v) \in[a_0,b_0]} \hrho_v(\bd_v).$$
\end{lemma}

\begin{proof}
The proof is very similar to the proof of Lemma~\ref{lem:line-key}, except that to upper bound the left hand side in the statement of the lemma, we only need to consider edges that are in the set $\cup_{\theta \in [a,b]} \delta(S_\theta)$. The condition in the lemma assures us that any node that is involved in $\delta(S_\theta)$ have to lie within the interval $[a_0,b_0]$. Thus, it is sufficient to consider the set of nodes $v: g(v) \in [a_0,b_0]$ in the integral on the right hand side. The proof is written out in detail in Sec.~\ref{sec:proof_line-refined}.
\end{proof}

Given a graph $G$ with edge lengths $\ell:E \rightarrow \mathbb{R}_+$,
a node $v$ and radius $r$, let $B^\ell_G(v,r) = \{u \mid \dist_\ell(v,u)
\le r\}$ denote the ball of radius $r$ around $v$ according to edge
lengths $\ell$. We omit $\ell$ and $G$ if they are clear from the context.
For a set of nodes $X \subseteq V$ we let $\vol(X) =
\sum_{v \in X} \hrho_v(\bd_v)$ denote the total contribution of the
nodes in $X$ to the objective function.

\begin{lemma}
  \label{lem:region-growing}
  Let $\delta < 1$ and suppose $\ell(e) < \frac{\delta}{2\log k}$ for
  all $e$. Then, for any given node $s$ and
  $k \ge 2$ there exists a $r \in [0, \delta)$ such that
  $\nu(\delta(B(s,r)) \le a \log k \cdot \frac{1}{\delta}
  (\vol(B(s,r)) + \vol(V)/k)$, with $a = 28$.
\end{lemma}

\begin{proof}
  For simplicity we assume here that $\log k$ is an integer multiple of $3$.
  Order the nodes in increasing order of distance from $s$: this
  produces a line embedding $g_s: V \rightarrow \mathbb{R}_+$.  For
  integer $i \ge 0$ define $r_i = \frac{i\cdot \delta}{2\log k}$.
  Define $\alpha_0 = \vol(V)/k$ and for $i \ge 1$ let $\alpha_i =
  \alpha_0 + \vol(B(s,r_i))$.

    Consider any $1 \leq j \leq 2 \log k $. We apply   Lemma~\ref{lem:line-refined} to the embedding $g_s$ and the interval
  $[r_{j-1},r_j]$; note that $\ell(e) < \frac{\delta}{2\log
    k}$ which implies that we can indeed apply the lemma. Also any edge  $ e \in \cup_{\theta \in [r_{j-1},r_j]}\delta(S_\theta) $  satisfies the property that $g(u) \in [r_{j-2},r_{j+1}]$ and $g(v) \in [r_{j-2},r_{j+1}]$ since $\ell(e) < \frac{\delta}{2\log k}$. Thus
  \begin{eqnarray}
      \int_{r_{j-1}}^{r_j} \nu(\delta(B(s,\theta)) d\theta & \le & 2\sum_{v: g_s(v) \in [r_{j-2},r_{j+1}]} \hrho_v(\bd_v) \nonumber \\
      & \le & 2(\alpha_{j+1} - \alpha_{j-2}) \label{eq:ub} \end{eqnarray}

   We claim that there is some $1 \le j <
  2\log k$ such that $\alpha_{j+1} \le 8 \alpha_{j-2}$. Suppose not, then $\alpha_{3i} > 8 \alpha_{3(i-1)}$ for all $1 \le i \le \frac{2\log k}{3}$. This implies that $\alpha_{3i} > 8^i \alpha_0 = 2^{3i} \alpha_0$. Therefore, with $i =  \frac{2\log k}{3}$, this implies that  $\alpha_{2\log k} > 2^{2 \log k} \frac{\vol(V)}{k}  >
  4\vol(V)$ which is impossible.

  Thus there exists a $j$ such that $\alpha_{j+1} \le 8\alpha_{j-2}$. Consider that $j$, equation \eqref{eq:ub} implies that
   \begin{eqnarray*}  \int_{r_{j-1}}^{r_j} \nu(\delta(B(s,\theta)) d\theta    & \le & 2(\alpha_{j+1} - \alpha_{j-2})\\
   & \leq & 2 (7 \alpha_{j-2}).  \end{eqnarray*}

  If we pick $r$ uniformly at random from the interval
  $[r_{j-1},r_j]$, where satisfies the above property, the expected cost of $\nu(\delta(B(s,r)))$ is
  $$ \frac{1}{r_j-r_{j-1}} \int_{r_{j-1}}^{r_j} \nu(\delta(B(s,\theta)) d\theta
  \le \frac{28 \log k}{\delta} \alpha_{j-2},$$
from the preceding inequality and the fact that $r_j-r_{j-1} =
\frac{\delta}{2\log k}$.  Hence there exists an $r \in [r_{j-1},r_j]$
such that $\nu(\delta(B(s,r))) \le \frac{28 \log k}{\delta}
\alpha_{j-2}$. Since $\alpha_{j-2} -\alpha_0 \le \vol(B(s,r))$, the lemma follows.
\end{proof}

Now we consider the following algorithm for finding a multicut from a given
fractional solution.
\begin{itemize}
\item Let $F \leftarrow \{e \mid \ell(e) \ge \frac{1}{4\log k}\}$.
\item $G' \leftarrow G[E\setminus F]$.
\item While (there exists a pair $s_it_i$ connected in $G'$) do
  \begin{itemize}
  \item Let $s_jt_j$ be a pair connected in $G'$.
  \item Via Lemma~\ref{lem:region-growing} with $\delta=1/2$ find $r <
    1/2$ such that $\nu(\delta_{G'}(B_{G'}(s_j,r))) \le 2a\log k \cdot
    (\vol(B_{G'}(s_j,r)) + \vol(V)/k)$.
  \item $F \leftarrow F \cup \delta_{G'}(B_{G'}(s_j,r))$.
    \item Remove the vertices $B_{G'}(s_j,r)$ and edges incident to them from $G'$.
  \end{itemize}
\item Output $F$ as the multicut.
\end{itemize}

\begin{lemma}
  The set of edges $F$ output by the algorithm is a feasible multicut for the given
  instance.
\end{lemma}
\begin{proof} (Sketch) One can prove this by induction on the number
  of steps in the while loop. We consider the first step.  The
  diameter of the ball $B_{G'}(s_j,r)$ is $2r < 1$ and hence the end
  points of any pair cannot both be inside this ball. We remove the
  edges $\delta(B_{G'}(s_j,r))$ and by the preceding observation there is
  no need to recurse on this ball. The algorithm recurses on the remaining
  graph $G' - B_{G'}(s_j,r)$, and by induction separates any pair with both end points
  in that graph.
\end{proof}

Now we argue about the cost of the set $F$ output by the algorithm.
Let $F_0 \leftarrow \{e \mid \ell(e) \ge \frac{1}{4\log k}\}$ be
the initial set of edges added to $F$ and let $F_i$ be the set of edges
added in the $i$'th iteration of the while loop.

\begin{lemma}
  $\nu(F_0) \le 8 \log k \cdot \sum_{v} \hrho_v(\bd_v)$.
\end{lemma}
\begin{proof}
  For $v \in V$ let $A_v = \{ e \in \delta(v) \cap F_0 \mid \ell(e,v)
  \ge \frac{1}{8\log k}\}$.  We can upper bound $\nu(F_0)$ by $\sum_v
  \rho_v(A_v)$ since the latter term counts each edge $uv \in F_0$ in
  at least one of $A_u$ and $A_v$ since $\ell(e,u) + \ell(e,v) =
  \ell(e) \ge \frac{1}{4\log k}$. From the definition of the Lov\'asz
  extension
  \[
  \hrho_v(\bd_v) = \int_0^1 \rho_v(\bd_v^\theta) d\theta \ge \int_{0}^{1/(8 \log k)}  \rho_v(\bd_v^\theta) d\theta \ge \frac{1}{8 \log k}  \rho_v(A_v),
  \]
  where we used non-negativity of $\rho_v$ for the first inequality above and
  monotonicity for the second.
\end{proof}

\begin{lemma}
  $\sum_{i\ge 1} \nu(F_i) \le 4a \log k \sum_v \hrho_v(\bd_v)$.
\end{lemma}
\begin{proof} (Sketch) From the algorithm description, $F_i =
  \delta(B_{G'}(s_j,r))$ for some terminal $s_j$ and radius $r < 1/2$
  where $G'$ is the remaining graph in iteration $i$. Moreover,
  $\nu(F_i) \le 2a \log k \cdot (\vol(B_{G'}(s_j,r)) +
  \vol(V)/k)$. Since the nodes in $B_{G'}(s_j,r)$ are removed from
  the graph, a node $u$ is charged only once inside a ball. Hence
$$\sum_i \nu(F_i) \le \sum_i 2a \log k \cdot \vol(V)/k + 2a \log k \sum_v \hrho_v(\bd_v) \le 4a \log k \sum_v \hrho_v(\bd_v),$$
since there are at most $k$ iterations of the while loop; each
iteration separates at least one pair.
  \end{proof}

Since $\nu$ is subadditive (see Lemma~\ref{lem:cut-subadditive})
$$\nu(F) \le \nu(F_0) + \sum_{i\ge 1} \nu(F_i) \le (8+4a)\log k \sum_v \hrho_v(\bd_v).$$
This finishes the proof of Theorem~\ref{thm:undir-multicut-gap}.

\section{Conclusions}
We considered multicommodity flows and cuts in polymatroidal networks
and derived flow-cut gap results in several settings. These results
generalize some existing results for the well-studied edge and
node-capacitated networks. We briefly mention two results that can be
obtained via the line embeddings technique that we did not include in
this paper. A multicommodity flow instance in an undirected network
$G=(V,E)$ is a product multicommodity flow instance if there there is
a non-negative weight function $\pi : V \rightarrow \mathbb{R}_+$ and
the demand $D_{uv}$ between $u$ and $v$ is $\pi(u)\cdot \pi(v)$. The
associated cut problem is interesting because it corresponds to
finding sparse separators in graphs which in turn can be used to find
balanced separators; these have several applications. It was shown in
\cite{KPR} that in edge-capacitated undirected planar networks, the
flow-cut gap for product multicommodity flow instances is $O(1)$ (in
fact they showed this holds for any class of graphs that excludes a
fixed graph as a minor). Rabinovich \cite{Rabinovich03} showed that
the main technical theorem in \cite{KPR} also leads to a line
embedding theorem, and this was used in \cite{FeigeHL06} to show an
$O(1)$ flow-cut gap for product multicommodity flow instances in {\em
  node}-capacitated planar graphs. Our work here shows that this is
true for undirected planar polymatroidal networks. Arora, Rao and
Vazirani \cite{ARV} gave an $O(\sqrt{\log n})$-approximation, via a
semi-definite programming relaxation, for the sparsest cut problem in
an undirected edge-capacitated network. Note that this is not a
traditional flow-cut gap result since the SDP-based relaxation used is
strictly stronger than the dual of the multicommodity flow relaxation.
By interpreting the main technical result in \cite{ARV} as a line-embedding
theorem, \cite{FeigeHL06} obtained an $O(\sqrt{\log n})$-approximation
for sparsest cut in node-capacitated graphs; this can also be extended
to the polymatroidal setting.

Flow-cut gap questions for node-capacitated problems are less
well-understood than the corresponding questions for edge-capacitated
problems; line-embeddings provide a tool to obtain upper
bounds on the gap but they do not provide a tight characterization
as $\ell_1$-embeddings do for the edge-capacitated case. We hope
that polymatroidal networks and their applications to network
information flow provide a new impetus for understanding these
questions.

\appendix
\section{Proof of Lemma~\ref{lem:formulation-equivalence}}
\label{sec:formulations-equivalence}
\begin{lemma}
  \label{lem:formulation-equivalence}
  For a directed polymatroidal network, the dual of the maximum
  throughput flow problem is equivalent (in terms of value) to the
  program given by \eqref{eq:dir} (for undirected the program given by
  \eqref{eq:undir}).
\end{lemma}
\bpf
We will show the proof for the undirected case, the proof for the directed case is similar. The program for maximum throughput flow is given by:
\begin{eqnarray*}
\max \quad \sum_{i} \sum_{p \in \mcP_{\st}} f(p)  \\
& \text{s.t.} & \\
\sum_{{e}:{e} \in S} \sum_{p: {e} \in p} f(p)  & \leq & \rho_v (S) \quad \forall S \subseteq \delta(v) \ \forall v \in V \\
f(p) & \geq & 0 \quad  \forall p \in \mc{P}_{\left ( s_i, t_i\right  )}, \forall i = 1\ldots k.
\end{eqnarray*}

The dual of the flow linear program can now be written. Let the dual
variables $d_v(S_v) $ correspond to the non-trivial constraint in the
above linear program. Then the dual linear program is:

\begin{eqnarray*}
\mc{P}_d := \min\quad \sum_{v \in V} \sum_{S \subseteq \delta(v) } d_v (S) \rho_v(S)  \\
& \text{s.t.} & \\
\sum_{e = uv: e \in p } \left( \sum_{S \subseteq \delta(u): e \in S} d_u(S) + \sum_{S \subseteq \delta(v): e \in S} d_v(S) \right)  & \geq & 1 \quad \forall p \in \mathcal{P}_{\st} \text{ where } \ \ e = uv \\
d_u(S) & \geq & 0 \ \forall u \in V\ \forall S \subseteq \delta (u).
\end{eqnarray*}

This can be rewritten equivalently as
\begin{eqnarray*}
\mc{P}_d & := & \min\quad \sum_{v \in V} \sum_{S \subseteq \delta(v) } d_v (S) \rho_v(S)  \\
& \text{s.t.} & \\
\ell(e) & := & \left( \sum_{S \subseteq \delta(u): e \in S} d_u(S) + \sum_{S \subseteq \delta(v): e \in S} d_v(S) \right) \\ \dist_{\ell}(s_i,t_i) & \geq & 1 \quad 1 \leq i \leq k   \\
d_u(S) & \geq & 0 \ \forall u \in V\ \forall S \subseteq \delta (u).
\end{eqnarray*}

Let us define new variables $\ell(e,u)$, $\ell(e,v)$ for each edge $e = uv$, and rewrite the linear program:
\begin{eqnarray*}
\min\quad \sum_{v \in V} \sum_{S \subseteq \delta(v) } d_v (S) \rho_v(S)  \\
& \text{s.t.} & \\
\ell(e) & := & \ell(e,u) + \ell(e,v), \text{ where } \ e = uv \\
\ell(e,u)&  =  &\sum_{S \subseteq \delta(u): e \in S} d_u(S) \ \forall e \in E, e = uv\\
\ell(e,v) & = & \sum_{S \subseteq \delta(v): e \in S} d_v(S) \quad \forall e \in E, e = uv\\
\dist_{\ell}(s_i,t_i) & \geq & 1 \quad 1 \leq i \leq k   \\
d_u(S) & \geq & 0 \\
\ell(e,u), \ell(e,v) & \geq & 0 \quad \forall u \in V\ \forall S \subseteq \delta (u).
\end{eqnarray*}

The minimization is over the variables $\ell(e,u)$ and $d_v(S)$.
Observe for any fixed $v$ the variables $d_v(S), S\subseteq \delta(v)$
influence only the variable $\ell(e,v), e \in \delta(v)$. Hence, for
any $v$ and a fixed assignment set of values $\ell(e,v), e \in \delta(v)$
the optimal choice of variables $d_v(S), S\subseteq \delta(v)$ can be
obtained by solving the following linear program:
\begin{eqnarray*}
\min\quad \sum_{S \subseteq \delta(v) } d_v (S) \rho_v(S) \\
& \text{s.t.} & \\
\sum_{S \subseteq \delta(v): e \in S} d_v(S) & = & \ell(e,v) \quad  \forall e \in E, e = uv\\
d_u(S) & \geq & 0 ,\quad S \subseteq \delta(v), \ \forall v \in V.
\end{eqnarray*}
Recalling the definition of the convex closure of a function,
one sees that the value of the above linear program is equal to
$\trho_v(\bd_v)$; note that for polymatroids we can drop the
constraint $\sum_S d_v(S) = 1$ in the linear program for the convex
closure. Since the convex closure is equval to the Lov\'asz extension
we obtain the desired equivalence of the formulations. \epf

\section{Proof of Lemma~\ref{lem:line-refined}} \label{sec:proof_line-refined}
\begin{lemma}
  Let $g:V \rightarrow [0,\beta]$ be a contraction, let $0 \leq a_0 \leq a < b \leq b_0 \leq \beta$ and $S_\theta = \{u \mid g(u) <  \theta\}$. Suppose for every edge $e=uv \in \cup_{\theta \in [a,b]} \delta(S_\theta)$, $g(u)$ and $g(v)$ are both in $[a_0,b_0]$. Then,
  $$\int_{a}^b \nu(\delta(S_\theta)) d\theta \le 2 \sum_{v: g(v) \in[a_0,b_0]} \hrho_v(\bd_v).$$
\end{lemma}

\begin{proof} Consider an edge $uv \in \delta(S_\theta)$ and for simplicity  assume $g(u) < g(v)$. The length of $e$ in the embedding is $\ell'(e) = |g(v)-g(u)| \le \ell(e)$. The edge $(u,v) \in \delta(S_\theta)$ iff $\theta$ is in the interval $[g(u),g(v)]$. Also by the conditions of the theory for every such $(u,v)$, $g(u) \in [a_0,b_0]$ and $g(v) \in [a_0,b_0]$. Note that the cost $\nu(\delta(S_\theta))$ is in general a complicated function to evaluate. We upper bound
  $\nu(\delta(S_\theta))$ by giving an explicit way to assign $e=uv$
  to either $u$ or $v$ as follows. Recall that in the relaxation $\ell(e) =
  \ell(e,u) + \ell(e,v)$ where $\ell(e,u)$ and $\ell(e,v)$ are the
  contributions of $u$ and $v$ to $e$. Let $r =
  \frac{\ell(e,u)}{\ell(e)}$ and let $\ell'(e,u) = r \ell'(e)$ and
  $\ell'(e,v) = (1-r)\ell'(e)$.  We partition the interval
  $[g(u),g(v)]$ into $[g(u),g(u)+\ell'(e,u))$ and
  $[g(u)+\ell'(e,u),g(v)]$; if $\theta$ lies in the former interval we
  assign $e$ to $u$, otherwise we assign $e$ to $v$.  This assignment
  procedures describes a way to upper bound $\nu(\delta(S_\theta))$
  for each $\theta$. Now we consider the quantity $\int_a^b
  \nu(\delta(S_\theta)) d\theta$ and upper bound it as follows.

  Consider a node $u$ and let $L_u = \{uv \in \delta(u) \mid g(v) <
  g(u) \}$ be the set of edges $uv$ that go from $u$ to the left of
  $u$ in the embedding $g$. Similarly $R_u = \{uv \in \delta(u) \mid
  g(v) \ge g(u) \}$. Note that $L_u$ and $R_u$ partition $\delta(u)$.
  Let $\bd'_u$ be the vector of dimension $|\delta(u)|$ consisting of
  the values $\ell'(e,u)$ for $e \in \delta(u)$.  We obtain $\bd^L_u$
  from $\bd'_u$ by setting the values for $e \in R_u$ to $0$ and
  similarly $\bd^R_u$ from $\bd'_u$ by setting the values for $e \in
  L_u$ to $0$. Since $0 \le \ell'(e,u) \le \ell(e,u)$ for each $e \in
  \delta(u)$ we see that $\bd'_u \le \bd_u$ and (component wise) and
  hence $\bd^L_u \le \bd_u$ and $\bd^R_u \le \bd_u$.  Since $\rho_u$
  is monotone we have that $\hrho_u(\bd^L_u) \le \hrho_u(\bd_u)$ and
  $\hrho_u(\bd^R_u) \le \hrho_u(\bd_u)$ (see
  Proposition~\ref{prop:monotone-sub}).

We claim that
$$\int_a^b \nu(\delta(S_\theta)) d\theta \le \sum_{u \in V: g(u) \in [a_0,b_0]} (\hrho_u(\bd^L_u) + \hrho_u(\bd^R_u)),$$
which would prove the lemma.

To see the claim consider some fixed $\theta$ and
$\nu(\delta(S_\theta))$. Fix a node $u$ and consider the edges in
$\delta(u) \cap S_\theta$ assigned to $u$ by the procedure we
described above; call this set $A_{\theta,u}$. First assume that
$\theta < g(u)$. Then the edges assigned to $u$ by the procedure,
denoted by $A_{\theta,u} = \{ e \in L_u \mid \theta > g(u) -
\ell'(e,u) \}$. Similarly, if $\theta > g(u)$, $A_{\theta,u} = \{ e
\in L_u \mid \theta < g(u) + \ell'(e,u) \}$. From these definitions we have
\begin{eqnarray*} \nu(\delta(S_\theta)) & \leq & \sum_{u \in V: g(u) \in [a_0,b_0]}  \rho_u(A_{\theta,u}) \\
\Rightarrow \int_a^b \nu(\delta(S_\theta)) d\theta & \le & \sum_{u \in V: g(u) \in [a_0,b_0]} \int_a^b \rho_u(A_{\theta,u})d\theta \\
 & \le & \sum_{u \in V: g(u) \in [a_0,b_0]} \int_0^{\beta} \rho_u(A_{\theta,u})d\theta. \end{eqnarray*}

For a fixed node $u$,
\begin{eqnarray*}
\int_0^{\beta} \rho_u(A_{\theta,u})d\theta & = &   \int_0^{g(u)} \rho_u(A_{\theta,u})d\theta + \int_{g(u)}^{\beta} \rho_u(A_{\theta,u})d\theta
\end{eqnarray*}
Let $L_u = \{e_1,e_2,\ldots, e_h\}$ where $0 \le \ell'(e_1,u) \le
\ell'(e_2,u) \le \ldots \le \ell'(e_h,u)$. Then
$$\int_a^{g(u)} \rho_u(A_{\theta,u})d\theta = \sum_{j=1}^h (\ell'(e_j,u)-\ell'(e_{j-1},u))\rho(\{e_1,e_2,\ldots,e_j\})$$
The right hand side of the above, is by construction and the
definition of the Lov\'{a}sz extension, equal to $\hrho_u(\bd^L_u)$.
Similarly, $\int_{g(u)}^\beta \rho_u(A_{\theta,u})d\theta =
\hrho_u(\bd^R_u)$.
\end{proof}

\section{ Proof of
  Theorem~\ref{thm:bi-partition} \label{app:bi-partition}}

We recall the statement of Theorem~\ref{thm:bi-partition}.

\medskip
\noindent
{\bf Theorem.} {\em Given any edge cut $F$ for a
multicommodity flow instance in an undirected polymatroidal network
$G=(V,E)$, there exists a bi-partition cut $\delta_G(S)$ whose sparsity is
atmost $2$ times the sparsity of $F$. Furthermore this factor
is tight. }

\bpf Let $V_1,V_2,\ldots,V_h$ be the connected components of $G$
induced by the removal of the edge-cut $F$. Let $D(F)$ be the total
demand separated by $F$. We show a cut $F' = \delta_G(S)$ such that
$D(F') \ge D(F)/2$ and $F' \subseteq F$.

We obtain $F'$ as follows. Construct an undirected graph with nodes
$v_1,...,v_h$, corresponding to the sets $V_1,\ldots,V_h$.  For each
pair $v_i,v_j$ we add an edge $v_iv_j$ with weight $w(v_iv_j)$ equal
to the total demand of all pairs with one end point in $V_i$ and the
other end point in $V_j$. Note that the total weight of all edges is
equal to $D(F)$. It is well-known that in any undirected weighted
graph there is a partition of the nodes into $A$ and $A^c$ (the
complement of $A$) such that the total weight of edges crossing the
partition is at least half the weight of all the edges in the graph (a
random partition guaratees this in expectation and gives a simple
$1/2$-approximation to the NP-Hard maximum-cut problem). Let $A$ be
such a cut in $H$; we have $w(\delta_H(A)) \ge D(F)/2$.  Now consider the
set of nodes $S = \cup_{v_i \in A} V_i$ in $G$ and the corresponding
cut $F' = \delta_G(S)$. It follows that $D(F') \ge D(F)/2$. Moreover
$F' \subseteq F$ and by monotonicity of $\nu(\cdot)$, $\nu(F') \le \nu(F)$.
Hence,
$$
\frac{\nu(F')}{D(F')} \leq 2 \frac{\nu(F)}{D(F)},
$$
which implies that the sparsity of $F'$ is at most twice that of $F$.
By construction $F'$ is a vertex bi-partition cut.

To see that the factor of $2$ is tight, consider a polymatroidal
network $G=(V,E)$ which is a star on $n$ nodes $\{v_0,
v_1,\ldots,v_{n-1}\}$ with center $v_0$; the edges are
$e_1,\ldots,e_{n-1}$ where $e_i = v_0v_i$.  Assume $n$ is even. The
only capacity constraint is a polymatroidal constraint at node $v_0$,
which constrains the total capacity of every subset of
$\{e_1,...,e_{n-1}\}$ by a value of $1$. The demand graph is a
complete graph on the nodes with each pair having a unit demand.
Now consider an edge cut $F$ which removes all the edges: $\nu(F)=1$
and $D(F) = \binom{n}{2}$, so the sparsity is $\frac{2}{n(n-1)}$,
whereas any bi-partition cut $F'=\delta(S)$ also has value $\nu(F')=1$ and
$D(F') = |S||S|^c$, which means the sparsity is minimized with $|S| =
\frac{n}{2}$ and is given by $\frac{4}{n^2}$. The ratio of the two sparsities
is $2(1-\frac{1}{n})$ and approaches $2$ as $n \rightarrow \infty$.
\epf

\end{document}